\documentclass[11pt,a4paper]{amsart}

\usepackage{amsfonts,amssymb,amsmath,eucal,pinlabel,array,hhline}
\usepackage[all]{xy}
\usepackage{tabulary}
\usepackage{fancyhdr}

\newtheorem{theorem}{Theorem}[section]
\newtheorem*{thmintro}{Theorem}
\newtheorem{lemma}[theorem]{Lemma}
\newtheorem{proposition}[theorem]{Proposition}

\theoremstyle{definition}\newtheorem{definition}{Definition}
\theoremstyle{remark}
\theoremstyle{remark}\newtheorem{remark}[theorem]{Remark}

\newenvironment{romanlist}
        {\begin{enumerate}
        }
        {\end{enumerate}}

\newcounter{ticklistc}
\newenvironment{ticklist}
    	{\setcounter{ticklistc}{0}
	 \begin{list}{--}
	{\usecounter{ticklistc}}}{\end{list}}

\newcommand{\Z}{\mathbb Z}
\newcommand{\R}{\mathbb R}
\newcommand{\C}{\mathbb C}

\newcommand{\EE}{\mathbb E}
\newcommand{\F}{\mathcal F}

\renewcommand{\L}{\mathcal L}
\renewcommand{\S}{\mathcal S}
\newcommand{\Hom}{\mathrm{Hom}}

\newcommand{\G}{\Gamma}
\newcommand{\SI}{\Sigma}

\newcommand{\Arf}{\mathrm{Arf}}
\renewcommand{\i}{\mathit{(i)}}
\newcommand{\ii}{\mathit{(ii)}}
\newcommand{\iii}{\mathit{(iii)}}

\begin{document}

\title{The critical Ising model via Kac-Ward matrices}

\author{David Cimasoni}
\address{Section de math\'ematiques, 2-4 rue du Li\`evre, 1211 Gen\`eve 4, Switzerland}
\email{David.Cimasoni@unige.ch}

\begin{abstract}
The Kac-Ward formula allows to compute the Ising partition function on any finite graph $G$ from the determinant of $2^{2g}$ matrices, where $g$ is the genus of a surface
in which $G$ embeds. We show that in the case of isoradially embedded graphs with critical weights, these determinants have quite remarkable properties.
First of all, they satisfy some generalized Kramers-Wannier duality: there is an explicit equality relating the determinants associated to a graph and to its dual graph.
Also, they are proportional to the determinants of the discrete critical Laplacians on the graph $G$, exactly when the genus $g$ is zero or one.
Finally, they share several formal properties with the Ray-Singer $\overline\partial$-torsions of the Riemann surface in which $G$ embeds.
\end{abstract}

\maketitle


\pagestyle{myheadings}
\markboth{D. Cimasoni}{The critical Ising model via Kac-Ward matrices}


\section{Introduction}
\label{sec:intro}

Most of the exact results for the two-dimensional Ising model rely on the so-called {\em Pfaffian method\/}. The idea, due independently to Hurst-Green~\cite{H-G}, Kasteleyn~\cite{Ka1} and
Fisher~\cite{Fi1,Fi2}, is to associate to the given graph $G$ an auxiliary graph $\G_G$ such that the dimer partition function on $\G_G$ is equal to the Ising partition function on $G$. The dimer
model technology can then be applied to solve the Ising model on $G$. In particular, if $G$ can be embedded in an orientable surface of genus $g$, then so can $\G_G$, and its dimer partition
function can be computed as an alternated sum of the Pfaffians of $2^{2g}$ well-chosen skew-adjacency matrices~\cite{G-L,Tes,D96,C-RI}. Even though this method is very standard at least for
planar graphs (see for example the classical book~\cite{MCW}), it is not a very natural one. First of all, this so-called {\em Fisher correspondence\/} $G\mapsto\G_G$ is by no means unique:
for example, each of the articles~\cite{Ka1,Fi2,D96,BdT1,M-L,Cim2} contains a different version of it. Secondly, whatever the chosen correspondence, virtually all the geometric and combinatorial
properties of $G$ will get lost when passing to $\G_G$ -- the only obvious exception being the genus. This will not be a problem if one is studying a topological class of graphs, such as all planar
graphs, or all finite graphs of a given genus. However, if one is interested in a geometric class -- as we will be -- the Pfaffian method leads to unnecessary complications.

There is another combinatorial method to solve the two-dimensional Ising model which, although naturally related to the Pfaffian method (see~\cite[Subsection 4.3]{Cim2}), is in our opinion
much more natural. It is due to Kac-Ward~\cite{KW} (even though a rigorous proof awaited
many years~\cite{DZMSS}), was originally formulated for planar graphs, and recently extended to any finite graph~\cite{Loe,Cim2}. With this method, no auxiliary graph is needed: the Ising partition
function on a finite graph $G$ is computed as an alternated sum of the square roots of the determinants of $2^{2g}$ {\em Kac-Ward matrices\/} naturally associated to the graph $G$.

In this paper, we initiate the study of the critical Ising model on graphs of {\em arbitrary\/} genus, and we do so using these Kac-Ward matrices. What we mean by ``critical Ising model" will
be formally defined and thoroughly motivated in Section~\ref{sec:def}, but for now, we will content ourselves with an informal description. Consider a finite number of planar rhombi
$\{R_e\}_{e\in E}$ of equal side length, each rhombus $R_e$ having a fixed diagonal $e$ and corresponding half-rhombus angle $\theta_e\in(0,\pi/2)$, as illustrated below.

\begin{figure}[htb]
\labellist\small\hair 2.5pt
\pinlabel {$e$} at 113 68
\pinlabel {$\theta_e$} at 45 77
\endlabellist
\centerline{\psfig{file=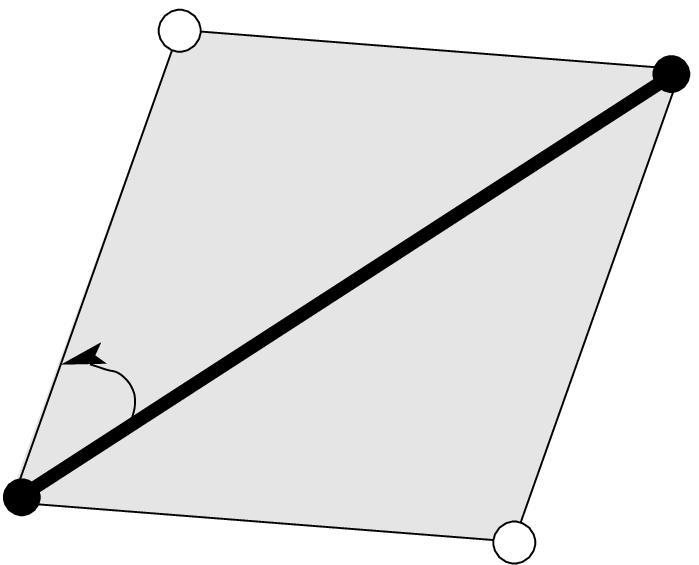,height=2cm}}
\end{figure}

Paste these rhombi together along their sides so that
extremities of diagonals are glued to extremities of diagonals. The result is a graph $G$ with edge set $E(G)=E$ embedded in a flat surface $\SI$ with so-called {\em cone-type singularities\/}
in the vertex set $V(G)$ of $G$, and in the middle of the faces of $G\subset\SI$. We shall say that $G$ is {\em isoradially embedded\/} in the flat surface $\SI$ (see Definition~\ref{def:crit}).
Using the high-temperature representation, we define the partition function for the {\em critical Ising model\/} on such a graph $G$ by
\[
Z(G,\nu)=\sum_{\gamma\in\mathcal{E}(G)}\prod_{e\in\gamma}\nu_e,
\]
where $\mathcal{E}(G)$ denotes the set of even subgraphs of $G$, and the {\em critical weights\/} $\nu_e$ are given by
$\nu_e=\tan(\theta_e/2)$. This is a natural generalization of the critical Z-invariant Ising model~\cite{Bax2}, which corresponds
to the special case where $\SI$ is a domain in the (flat) plane. Note that only a specific class of planar graphs admit a Z-invariant Ising model (see~\cite{K-S,CS}).
On the other hand, any finite graph can be isoradially embedded in a flat surface as explained above, and therefore admits a critical Ising model (Proposition~\ref{prop:real}).

Since its introduction by Baxter, the critical Z-invariant Ising model has been extensively studied, as well as its analog on the flat torus (see for example the papers~\cite{BdT1,BdT2}, where
Boutillier and de Tili\`ere make use of the Pfaffian method). On the other hand, very little is known about the critical Ising model on graphs in higher genus. It is widely believed that such models
should be discrete analogs of some conformal field theory~\cite{AMV}, but in genus $g\ge 2$, such statements are only supported by numerical experiments on very specific examples~\cite{CSM1,CSM2}.
It is our belief that, in order to try to tackle such outstanding conjectures, the Kac-Ward method will prove useful.

Before doing so, we need to settle some fundamental questions about the critical Ising model on graphs of arbitrary genus and the associated Kac-Ward matrices,
and this is exactly what the present paper is about.

In the general case of an arbitrary graph embedded in a topological surface, the (generalized) Kac-Ward matrices can be quite complicated (see Definition~1 in~\cite{Cim2}).
The first nice surprise is that for a weighted graph
$(G,x)$ embedded in a flat surface, the corresponding Kac-Ward matrices take a remarkably simple form. For each homomorphism $\varphi$ from the fundamental group $\pi_1(\SI)$ of $\SI$ to the group
$S^1$, we get a {\em $\varphi$-twisted Kac-Ward matrix\/} of order $2|E(G)|$ whose determinant we denote by $\tau^\varphi(G,x)$ (see Definition~\ref{def:KW}).
For a specific type of $\varphi$'s -- the ones that belong to the set $\S$ of
{\em discrete spin structures\/} on $G\subset\SI$ -- $\tau^\varphi(G,x)$ turns out to be the square of a polynomial in the variables $\{x_e\}_e$. The main theorem of~\cite{Cim2} then easily implies
the following result (see Theorem~\ref{thm:Arf} for the complete statement).

\begin{thmintro}
If all cone angles are odd multiple of $2\pi$, then the Ising partition function on the weighted graph $(G,x)$ is given by
\[
Z(G,x)=\frac{1}{2^g}\sum_{\lambda\in\S}(-1)^{\Arf(\lambda)}\tau^\lambda(G,x)^{1/2},
\]
where $g$ is the genus of $\SI$ and $\Arf(\lambda)\in\Z_2$ the Arf invariant of the discrete spin structure $\lambda$.
\end{thmintro}

Let us now turn to our main results. In a nutshell, we show that when the weight system is critical the corresponding Kac-Ward determinants $\tau^\varphi(G,\nu)$ exhibit
several remarkable properties.
 
Firstly, these determinants turn out to admit a relatively simple combinatorial interpretation, as described in Proposition~\ref{prop:tech}.
Furthermore, they satisfy the following duality property.

\begin{thmintro}
Let $G$ be a graph isoradially embedded in a flat surface $\SI$, and let $\nu$ be the critical weight system on $G$. The dual graph $G^*$ is also isoradially embedded in $\SI$, and therefore admits
a critical weight system $\nu^*$. If all cone angles are odd multiples of $2\pi$, then for any $\varphi$,
\[
2^{|V(G^*)|}\hskip-2pt\prod_{e^*\in E(G^*)}(1+\cos(\theta_{e^*}))\,\tau^\varphi(G^*,\nu^*)=2^{|V(G)|}\hskip-2pt\prod_{e\in E(G)}(1+\cos(\theta_{e}))\,\tau^\varphi(G,\nu).
\]
\end{thmintro}
This can be interpreted as a generalization of the celebrated Kramers-Wannier duality~\cite{K-W} from the case of planar graphs to the case of graphs of arbitrary genus.

Our final result relates the Kac-Ward matrices to some {\em a priori\/} totally different operator. Given a weighted graph $(G,x)$ and a homomorphism $\varphi\colon\pi_1(G)\to S^1$, the
associated {\em discrete Laplacian\/} on $G$ is the operator $\Delta^\varphi=\Delta^\varphi(G,x)$ acting on $f\in\C^{V(G)}$ by
\[
(\Delta^\varphi f)(v)=\sum_{e=(v,w)}x_e\,\left(f(v)-f(w)\varphi(e)\right),
\]
the sum being over all oriented edges $e$ of the form $(v,w)$. Note that if $G$ is a (planar) isoradial graph, then the corresponding critical weights are given by $c_e=\tan(\theta_e)$
(see~\cite{Ken}). We prove:

\begin{thmintro}
Let $G\subset\Sigma$ be an isoradially embedded graph in a flat surface, and let us assume that all cone angles $\vartheta_v$
of singularities $v\in V(G)$ are odd multiples of $2\pi$. If the genus of $\SI$ is $0$ or $1$, then for any $\varphi\colon\pi_1(\SI)\to S^1$,
\[
\tau^\varphi(G,\nu)=(-1)^N\,2^{-\chi(G)}\prod_{e\in E(G)}\frac{\cos(\theta_e)}{1+\cos(\theta_{e})}\,\det\Delta^\varphi(G,c),
\]
where $N$ is the number of vertices $v\in V(G)$ such that $\vartheta_v/2\pi$ is congruent to $3$ modulo $4$, and $\chi(G)=|V(G)|-|E(G)|$.
On the other hand, the functions $\tau^\varphi(G,\nu)$ and $\det\Delta^\varphi(G,c)$ are never proportional if the genus of $\SI$ is greater or equal to two.
\end{thmintro}

Via the method developed in~\cite[Subsection 4.3]{Cim2}, this theorem can be interpreted as a wide-reaching generalization of the main result of \cite{BdT1,BdT2} which was obtained via
the Pfaffian method. In our opinion, this is a good example of how simpler and more natural a proof can get, when the Kac-Ward method is used instead of the Pfaffian one.
(See Remark~\ref{rem:BdT} below for a more detailed comparison.)
In the case of the flat torus, the theorem above implies a relation between the free energy of the critical Z-invariant Ising model on $G$
and the normalized determinant of the critical discrete Laplacian on $G$. The former quantity was computed by Baxter~\cite{Bax2}, the latter by Kenyon~\cite{Ken}, and our equality allows to obtain
any of these two results as a corollary of the other one.

\medskip

The paper is organized as follows. In Section~\ref{sec:def}, we define our model: the Ising model on graphs isoradially embedded in flat surfaces (Definition~\ref{def:crit}), with critical weights
(Definition~\ref{def:nu}). In Section~\ref{sec:KW}, we introduce the $\varphi$-twisted Kac-Ward matrices for graphs in flat surfaces (Definition~\ref{def:KW}), and we show how they can be used
to compute the Ising partition function (Theorem~\ref{thm:Arf}). Section~\ref{sec:crit} deals with the case of isoradially embedded graphs with critical weights, and contains our main results.
We start with the combinatorial interpretation for the Kac-Ward determinants with critical weights (Proposition~\ref{prop:tech}). Then, we prove the equality relating the Kac-Ward determinants of
dual isoradially embedded graphs (Theorem~\ref{thm:duality}). Also, we relate the Kac-Ward determinants with the determinant of the critical discrete Laplacian on $G$ (Theorem~\ref{thm:Delta}).
In a last paragraph, we explain how the Kac-Ward determinant can be understood as a discrete version of the $\overline\partial$-torsion of the underlying Riemann surface (Subsection~\ref{sub:RS}).

\subsection*{Acknowledgments}
This research was supported by the European Research Council AG CONFRA and by the Swiss NSF. Part of this paper was done at the ETH in Zurich, and it is a pleasure to thank the Mathematics
Department of the ETHZ for providing such an excellent working environment. The author also wishes to thank C\'edric Boutillier and Hugo Duminil-Copin for comments on an earlier
version of the manuscript, and Martin Loebl for valuable discussions.


\section{The critical Ising model on isoradial graphs}
\label{sec:def}

The aim of this section is to explain the setup of the model that we will be studying in this paper: the critical Ising model on graphs isoradially embedded in a flat surface
(Definitions~\ref{def:crit} and~\ref{def:nu}). To motivate this definition, we start by recalling what is meant by high and low-temperature expansions for the Ising model, leading to
the Kramers-Wannier duality argument (Subsection~\ref{sub:KW}). The models on which such an argument can be applied are called Z-invariant Ising models (Subsection~\ref{sub:Z}).
They are all defined on planar (or toric) graphs, but a generalization of these models to surfaces of arbitrary genus then naturally leads to our definition (Subsection~\ref{sub:flat}).

\subsection{Kramers-Wannier duality}
\label{sub:KW}

Let $G$ be a finite graph with vertex set $V(G)$ and edge set $E(G)$. A {\em spin configuration\/} on $G$ is a map $\sigma\colon V(G)\to\{-1,+1\}$. Any positive edge weight system
$J=(J_e)_{e\in E(G)}$ on $G$ determines a probability measure on the set $\Omega(G)$ of such spin configurations by
\[
P(\sigma)=\frac{1}{Z^J(G)}\exp\Big(\sum_{e=(u,v)\in E(G)}J_e\sigma_u\sigma_v\Big),
\]
where
\[
Z^J(G)=\sum_{\sigma\in\Omega(G)}\exp\Big(\sum_{e=(u,v)\in E(G)}J_e\sigma_u\sigma_v\Big)
\]
is the {\em partition function\/} of the {\em Ising model on $G$ with coupling constants $J$\/}.

As observed by van der Waerden \cite{vdW}, the identity
\[
\exp(J_e\sigma_u\sigma_v)=\cosh(J_e)(1+\tanh(J_e)\sigma_u\sigma_v)
\]
allows to express this partition function as
\begin{align*}
Z^J(G)&=\Big(\prod_{e\in E(G)}\cosh(J_e)\Big)\sum_{\sigma\in\Omega(G)}\prod_{e=(u,v)\in E(G)}(1+\tanh(J_e)\sigma_u\sigma_v)\\
&=\Big(\prod_{e\in E(G)}\cosh(J_e)\Big)2^{|V(G)|}\sum_{\gamma\in\mathcal{E}(G)}\prod_{e\in\gamma}\tanh(J_e),\\
\end{align*}
where $\mathcal{E}(G)$ denotes the set of even subgraphs of $G$, that is, the set of subgraphs $\gamma$ of $G$ such that every vertex of $G$ is adjacent to an even number of edges of $\gamma$.
This is called the {\em high-temperature expansion\/} of the partition function.

Let us now assume that the graph $G$ is planar, and let $G^*$ denote its dual graph. For any spin configuration $\sigma\in\Omega(G)$, consider the subgraph of $G^*$ given by
all edges $e^*\in E(G^*)$ dual to $e=(u,v)$ with $\sigma_u\neq\sigma_v$. Clearly, this is an even subgraph of $G^*$. Furthermore, since $G$ is planar, this defines a surjective map
$\Omega(G)\to\mathcal{E}(G^*)$
such that $\sigma$ and $\sigma'$ have same image if and only if $\sigma'=-\sigma$. This leads to the following {\em low-temperature expansion\/} of the Ising partition function:
\begin{align*}
Z^J(G)&=2\sum_{\gamma^*\in\mathcal{E}(G^*)}\prod_{e\in E(G)}\exp(J_e)\prod_{e*\in\gamma^*}\exp(-2J_e)\\
&=2\Big(\prod_{e\in E(G)}\exp(J_e)\Big)\sum_{\gamma^*\in\mathcal{E}(G^*)}\prod_{e*\in\gamma^*}\exp(-2J_e).\\
\end{align*}
Thus, if we assign weights $J$ to $E(G)$ and $J^*$ to $E(G^*)$ in such a way that $\tanh(J_{e^*})=\exp(-2J_e)$, or more symmetrically,
\[
\sinh(2J_e)\sinh(2J_{e^*})=1,
\]
we obtain that the partition functions $Z^J(G)$ and $Z^{J^*}(G^*)$ are proportional to each other.

In the case of the square lattice with constant weight system $J$, this is enough to determine the critical value $J_c$ of the coupling constant.
Indeed, assuming that the free energy of the model is analytic everywhere except at a single point, this point must be equal to $J_c$ {\em and\/} to $J_c^*$, since the square lattice is self-dual.
The equality $J_c=J^*_c$ then leads to the explicit value $J_c=\log\sqrt{1+\sqrt{2}}$. 

\begin{figure}[Htb]
\labellist\small\hair 2.5pt
\pinlabel {$v$} at 145 100
\endlabellist
\centerline{\psfig{file=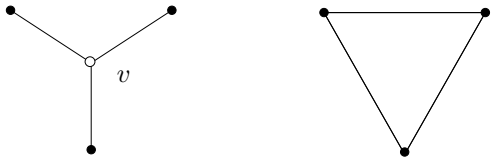,height=2cm}}
\caption{The star-triangle transformation at a vertex $v$.}
\label{fig:star}
\end{figure}

\subsection{The Z-invariant Ising models}
\label{sub:Z}

This beautifully simple argument, the celebrated {\em Kramers-Wannier duality\/} \cite{K-W}, is not sufficient to determine the critical value of the coupling constant on a graph that is not self-dual:
it only relates this critical value to the one for the dual graph. However, for some planar graphs, this can be obtained ``with little additional labor" \cite{Wan}.

Let us start with the example of the hexagonal lattice $H$, and let $H'$ denote the graph obtained from $H$ by a {\em star-triangle transformation\/} at a vertex $v$ as illustrated in
Figure~\ref{fig:star}. Assume one can assign coupling constants $J$ to the edges of $H$ and $J'$ to the newly created edges of $H'$ in such a  way that $Z^{J'}(H')=R_v\,Z^J(H)$
for some function $R_v$ of the coupling constants of the edges around $v$. Since the triangular lattice $T$ can be obtained from $H$ by such transformations, it would follow that
$Z^K(T)=R\,Z^J(H)$ for some controlled $R$ and well-chosen coupling constants $K,J$. This, together with the Kramers-Wannier duality, would lead to an equality of the form $Z^K(T)=k\,Z^{K^*}(T)$,
with $K\mapsto K^*$ some involution and $k$ an explicit function of $K$ and $K^*$. Arguing as above, the critical points should be self-dual under this involution, leading to the equality
$k=1$ and an exact description of these critical points for the triangular and hexagonal lattices. 

This strategy of using invariance under the star-triangle transformation (or, {\em Z-invariance\/}) can be applied not only to the hexagonal-triangular lattices, but to a wide class of planar
graphs~\cite{Bax,Bax2}. It turns out that this class of graphs on which a Z-invariant Ising model can be defined coincides with the graphs that admit an {\em isoradial embedding\/}
in the plane \cite{K-S,CS}: this is an embedding such that each face is inscribed in a circle of radius one, with the circumcenter in the closure of the face. Furthermore, the corresponding
critical coupling constants admit a very simple geometric description: they are given by
\[
J_e=\frac{1}{2}\log\left(\frac{1+\sin\theta_e}{\cos\theta_e}\right),
\]
where $\theta_e\in(0,\pi/2)$ is the half-rhombus angle associated to the edge $e$, as illustrated in Figure~\ref{fig:theta}.
For example, the square lattice is isoradially embedded with all half-rhombus angles equal to $\theta=\pi/4$, leading to the critical coupling constant $J_c=\log\sqrt{1+\sqrt{2}}$ as above.
On the other hand, the triangular and hexagonal lattices are isoradially embedded with angles $\theta=\pi/6$ (resp. $\pi/3$), so the corresponding critical values are equal to
$J_c=\log\sqrt{\sqrt{3}}$ (resp. $\log\sqrt{2+\sqrt{3}}$).

This geometric description of the critical coupling constants becomes even nicer when using the high-temperature expansion, as we obtain the {\em critical weights\/}
\[
\nu_e=\tanh(J_e)=\tanh\left(\frac{1}{2}\log\left(\frac{1+\sin\theta_e}{\cos\theta_e}\right)\right)=\frac{1-\cos\theta_e}{\sin\theta_e}=\tan(\theta_e/2).
\]

\begin{figure}[Htb]
\labellist\small\hair 2.5pt
\pinlabel {$e$} at 180 135
\pinlabel {$\theta_e$} at 115 140
\endlabellist
\centerline{\psfig{file=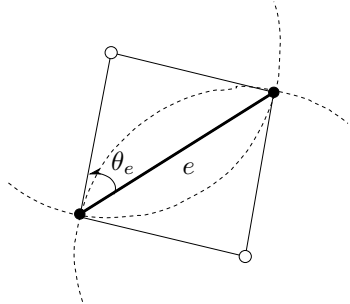,height=4cm}}
\caption{An edge $e$ of an isoradial graph, the associated rhombus, and the half-rhombus angle $\theta_e$.}
\label{fig:theta}
\end{figure}

\subsection{Isoradial graphs in flat surfaces}
\label{sub:flat}

In the present paper, we will study more general models where the graph $G$ is not assumed to be planar. The proper generalization of planar isoradiality is obtained by considering so-called
flat surfaces with cone-type singularities. Let us quickly recall their definition and main properties, referring to \cite{Tro} for further details.

Given a positive real number $\vartheta$, the space
\[
C_\vartheta=\{(r,t)\,:\,\hbox{$r\ge 0$, $t\in\R/\vartheta\Z$}\}/(0,t)\sim(0,t')
\]
endowed with the metric $\mathit{ds}^2=\mathit{dr}^2+r^2\mathit{dt}^2$ is called the
{\em standard cone of angle $\vartheta$\/}. Note that the cone without its tip is locally isometric to the Euclidean plane.
Let $\Sigma$ be a surface with a discrete subset $S$. A {\em flat metric on $\Sigma$ with cone-type singularities} of angles
$\{\vartheta_x\}_{x\in S}$ supported at $S$ is an atlas $\{\phi_x\colon U_x\to U'_x\subset C_{\vartheta_x}\}_{x\in S}$, where $U_x$ is an open neighborhood of $x\in S$, $\phi_x$ maps
$x$ to the tip of the cone $C_{\vartheta_x}$, and the transition maps are Euclidean isometries.

This seemingly technical definition should not hide the fact that these objects are extremely simple and natural: any such flat surface can be obtained by gluing polygons embedded
in $\R^2$ along pairs of sides of equal length. For example, a rectangle with opposite sides identified will define a flat torus with no singularity. On the other hand,
a regular $4g$-gon with opposite sides identified gives a flat surface of genus $g$ with a single singularity of angle $2\pi(2g-1)$. In general, the topology of the surface
is related to the cone angles by the following Gauss-Bonnet Formula: if $\SI$ is a closed flat surface with cone angles $\{\vartheta_x\}_{x\in S}$, then
\[
\sum_{x\in S}(2\pi-\vartheta_x)=2\pi\chi(\Sigma),
\]
where $\chi(\SI)$ is the Euler characteristic of $\SI$.

\begin{definition}
\label{def:crit}
A graph $G$ is {\em isoradially embedded in a flat surface $\Sigma$\/} if the following conditions are satisfied:
\begin{ticklist}
\item{$\Sigma$ is a compact orientable flat surface with cone-type singularities;}
\item{each edge of $G$ is a straight line in $\SI$;}
\item{each closed face $f$ of $G\subset\Sigma$ contains an element $x_f$ at distance $1$ from all vertices of $\partial f$;}
\item{a singularity of $\Sigma$ is either a vertex of $G$ or a vertex $x_f$ of the dual graph $G^*$, that is, the singular set is contained in $V(G)\cup V(G^*)$.}
\end{ticklist}
\end{definition}

Given an isoradially embedded graph $G\subset\Sigma$, each edge $e\in E(G)$ has an associated rhombus as illustrated in Figure~\ref{fig:theta}.
Therefore, the metric space $\Sigma$ should simply be understood as rhombi pasted together along their boundary edges.
This observation also leads to the following fact.

\begin{proposition}
\label{prop:real}
Any finite graph $G$ can be isoradially embedded in a flat surface.
\end{proposition}
\begin{proof}
Fix an arbitrary angle $\theta\in(0,\pi/2)$ and associate to each half-edge $\tilde{e}$ of $G$ the isosceles triangle illustrated below.
\begin{figure}[h]
\labellist\small\hair 2.5pt
\pinlabel {$\tilde{e}$} at 175 60
\pinlabel {$1$} at 320 110
\pinlabel {$1$} at 80 110
\pinlabel {$\theta$} at 175 110
\pinlabel {$\theta$} at 220 110
\endlabellist
\centerline{\psfig{file=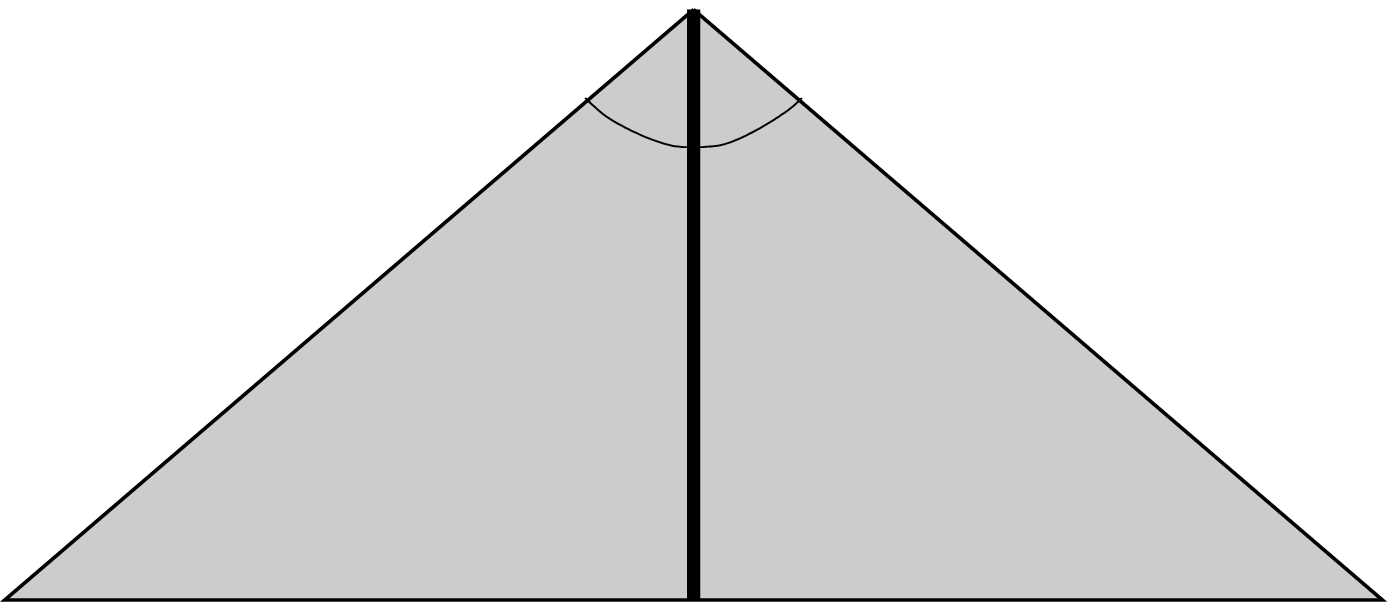,height=1.5cm}}
\end{figure}

For each vertex $v\in V(G)$, choose a cyclic ordering of the $d_v$ half-edges of
$G$ coming out of $v$ and build the associated star $\mathit{St}(v)$ as follows: glue together the corresponding $d_v$ triangles following the chosen cyclic ordering around $v$.
Note that $\mathit{St}(v)$ is a flat surface with one singularity of cone angle $2\theta d_v$, and that the cyclic ordering endows this surface with an orientation. For each edge $e=(v,w)\in E(G)$,
glue together the stars $\mathit{St}(v)$ and $\mathit{St}(w)$ along their boundary by pasting together the two triangles associated to $e$ in the unique way that is consistent with the
orientation of the stars. (If $e$ is a loop at $v$, just glue together the corresponding sides of the star $\mathit{St}(v)$.) By construction, $G$ is isoradially embedded in the resulting
metric space $\SI$, which is a compact oriented flat surface.
\end{proof}

Following the discussion of the previous subsections, we shall adopt the following terminology.

\begin{definition}
\label{def:nu}
Let $G$ be a graph isoradially embedded in a flat surface $\SI$.
The {\em critical weight\/} associated to the edge $e\in E(G)$ is defined by
\[
\nu_e=\tan(\theta_e/2),
\]
where $\theta_e\in(0,\pi/2)$ is the half-rhombus angle associated to the edge $e$.
The partition function for the {\em critical Ising model\/} on $G$ is given by
\[
Z(G,\nu)=\sum_{\gamma\in\mathcal{E}(G)}\prod_{e\in\gamma}\nu_e,
\]
where $\mathcal{E}(G)$ denotes the set of even subgraphs of $G$.
\end{definition}


\section{The Kac-Ward formula for graphs in flat surfaces}
\label{sec:KW}

In \cite{Cim2}, we gave a generalized Kac-Ward formula for the Ising partition function on any finite weighted graph $(G,x)$.
The aim of this section is to show that when the graph $G$ is embedded in a flat surface, the generalized Kac-Ward matrices take a particularly simple form --
whatever the weight system $x$ on $G$ is. In the next section, we will consider the case of isoradial graphs with critical weights.

\subsection{Kac-Ward matrices for graphs in flat surfaces}

Let us start with some general terminology and notation. Given a weighted graph $(G,x)$, let ${\EE}={\EE}(G)$ be the set of oriented edges of $G$.
Following~\cite{Ser}, we shall denote by $o(e)$ the origin of an oriented edge $e\in\EE$, by $t(e)$ its terminus, and by $\bar{e}$ the same edge with the opposite orientation.
By abuse of notation, we shall write $x_e=x_{\bar{e}}$ for the weight associated to the unoriented edge
corresponding to $e$ and $\bar{e}$.

Now, assume that $G$ is embedded in an orientable flat surface $\SI$ so that each edge of $G$ is a straight line, $\SI\setminus G$ consists of topological discs, and
the set $S$ of cone-type singularities is contained in $V(G)\cup V(G^*)$. As above, let $\vartheta_x$ denote the cone angle of the singularity $x\in S$.
Fix a unitary character $\varphi$ of the fundamental group of $\SI$, that is, an element of
\[
\Hom(\pi_1(\SI),U(1))=\Hom(\pi_1(\SI),S^1)=\Hom(H_1(\SI),S^1)=H^1(\Sigma;S^1).
\]
As $G\subset\SI$ induces a cellular decomposition of $\SI$, one can represent such a cohomology class by a cellular 1-cocycle, that we shall also denote by $\varphi$.
This is nothing but a map from the set ${\EE}$ of oriented edges of $G$ into $S^1$, such that $\varphi(\bar{e})=\overline{\varphi(e)}$ and $\varphi(\partial f)=\prod_{e\in\partial f}\varphi(e)=1$
for each face $f$ of $G\subset\SI$.

\begin{definition}
\label{def:KW}
Let $T^\varphi$ denote the $|{\EE}|\times|{\EE}|$ matrix defined by
\[
T^\varphi_{e,e'}=
\begin{cases}
\varphi(e)\,i\exp\left(-\frac{i}{2}\beta(e',\bar{e})\right)\,x_e& \text{if $t(e)=o(e')$ but $e'\neq \bar{e}$;} \\
0 & \text{otherwise,}
\end{cases}
\]
where $\beta(e',\bar{e})\in (0,\vartheta_v)$ denotes the angle from $e'$ to $\bar{e}$, as illustrated in the left part of Figure~\ref{fig:alpha}.
We shall call the matrix $I-T^\varphi$ the {\em $\varphi$-twisted Kac-Ward matrix\/} associated to the weighted graph $(G,x)$, and denote its determinant
by $\tau^\varphi(G,x)=\det(I-T^\varphi)$.
\end{definition}

\begin{figure}[Htb]
\labellist\small\hair 2.5pt
\pinlabel {$e$} at 600 110
\pinlabel {$\alpha(e,e')$} at 840 100
\pinlabel {$e'$} at 780 190
\pinlabel {$v$} at 170 70
\pinlabel {$v$} at 690 70
\pinlabel {$\beta(e',\bar{e})$} at 180 200
\pinlabel {$e'$} at 310 140
\pinlabel {$\bar{e}$} at 40 130
\endlabellist
\centerline{\psfig{file=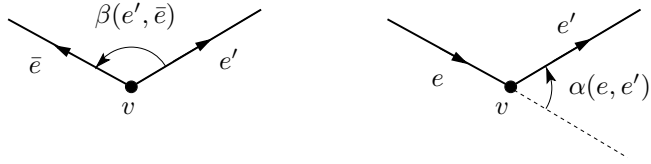,height=2cm}}
\caption{The angles $\beta(e',\bar{e})$ and $\alpha(e,e')=\pi-\beta(e',\bar{e})$.}
\label{fig:alpha}
\end{figure}

\begin{remark}
\label{rem:alpha}
If $t(e)=o(e')=v$ is not a singularity (that is, if the cone angle $\vartheta_v$ is equal to $2\pi$), then the complex number $i\exp\left(-\frac{i}{2}\beta(e',\bar{e})\right)$ is equal to
$\exp\left(\frac{i}{2}\alpha(e,e')\right)$, with $\alpha(e,e')\in (-\pi,\pi)$ the angle from $e$ to $e'$ illustrated in the right part of Figure~\ref{fig:alpha}. If $v$ does belong to $S$,
then the dotted line drawn there does not make sense anymore, hence the necessity to adopt this slightly less intuitive definition.
\end{remark}

\subsection{Discrete spin structures on flat surfaces}

The Kac-Ward matrices are particularly useful when $\varphi$ is a special type of 1-cocycle, namely a discrete spin structure.
In this paragraph, we shall recall the definition and main properties of these objects, slightly generalizing Section 3.1 of~\cite{Cim1}.

Loosely speaking, a {\em spin structure\/} on an oriented surface $\SI$ is a way to count parity of rotation numbers for closed curves in $\SI$. In the plane, there is a unique way to do so, and
therefore a unique spin structure. On a closed orientable surface of genus $g$ however, there are exactly $2^{2g}$ distinct spin structures.
More precisely, one can identify the set $\S(\SI)$ of spin structures on $\SI$ with the set of {\em quadratic forms\/} on $\SI$~\cite{Joh}: these are the maps $q\colon H_1(\SI;\Z_2)\to\Z_2$
such that $q(x+y)=q(x)+q(y)+x\cdot y$ for all $x,y\in H_1(\SI;\Z_2)$, where $x\cdot y$ denotes the intersection number of $x$ and $y$.
Note that the difference of two quadratic forms is a linear form. Therefore, this set admits a freely transitive action of the abelian group $H^1(\SI;\Z_2)$; in other words, it is an affine
$H^1(\SI;\Z_2)$-space. The {\em Arf invariant\/} of a spin structure is defined as the Arf invariant of the associated quadratic form, that is, the number $\Arf(q)\in\Z_2$ satisfying
\[
(-1)^{\Arf(q)}=\frac{1}{2^g}\sum_{\alpha\in H_1(\SI;\Z_2)}(-1)^{q(\alpha)}.
\] 

Coming back to flat surfaces, let us assume that $G$ is a graph embedded in a flat surface $\SI$ so that $\SI\setminus G$ consists of topological discs, and let $X$ denote the induced
cellular decomposition. We shall now explain how, in such a situation, it is possible to encode spin structures on $\SI$ by some cocycles $\lambda\in Z^1(X;S^1)$.
Let us assume that all cone angles $\vartheta_x$ are positive multiples of $2\pi$, i.e. that $\SI$ has trivial local holonomy. Then, the holonomy defines an element
$\mathrm{Hol}$ of $\mathrm{Hom}(\pi_1(\Sigma),S^1)=H^1(\Sigma;S^1)=H^1(X;S^1)$. We shall call a cocycle $\kappa\in Z^1(X;S^1)$ such that $[\kappa^{-1}]=\mathrm{Hol}$ a
{\em discrete canonical bundle\/} over $\Sigma$. Note that such a cocycle $\kappa$ is very easy to determine. Indeed, it is always possible to represent $\Sigma$ as planar polygons
$P$ with boundary identifications. Furthermore, these polygons can be chosen so that $G$ intersects $\partial P$ transversally, except at possible singularities in $S\cap V(G)$.
Define $\kappa$ by
\[
\kappa(e)=
\begin{cases}
1 & \text{if $e$ is contained in the interior of $P$;} \\
\exp(-i\theta) & \text{if $e$ meets $\partial P$ transversally,}
\end{cases}
\]
where $\theta$ denotes the angle between the sides of $\partial P\subset\C$ met by the edge $e$.  If $S\cap V(G)$ is empty, this
defines completely a natural choice of discrete canonical bundle $\kappa$. Otherwise, the partially defined $\kappa$ above can be extended to a cocycle
yielding a discrete canonical bundle.

Mimicking the continuous case (in the version developed by Atiyah~\cite{Ati}), let us define a {\em discrete spin structure\/} on $\Sigma$ as any cellular 1-cocycle $\lambda\in Z^1(X;S^1)$ such that
$\lambda^2=\kappa$. Two discrete spin structures will be called {\em equivalent\/} if they are cohomologous. The set $\S(X)$ of equivalent classes of discrete spin structures on
$\Sigma$ is then given by
\[
\S(X)=\{[\lambda]\in H^1(X;S^1)\,|\,[\lambda]^2=[\kappa]\}.
\]
Note that if the flat surface $\Sigma$ has trivial holonomy, then $[\kappa]$ is trivial, so the set $\S(X)$ is equal to the $2g$-dimensional vector space $H^1(\Sigma;\Z_2)$.
In general, $G\subset\Sigma$ can be described via planar polygons as explained earlier. In such a case, and assuming that the singular set $S\cap V(G)$ is empty,
a discrete spin structure is given by
\[
\lambda(e)=
\begin{cases}
1 & \text{if $e$ is contained in the interior of $P$;} \\
\exp(-i\theta/2) & \text{if $e$ meets $\partial P$,}
\end{cases}
\]
where $\exp(-i\theta/2)$ denotes one of the square roots of the angle between the sides of $\partial P\subset\C$ met by the edge $e$.

One easily checks that the set $\S(X)$ is an affine $H^1(\Sigma;\Z_2)$-space. Furthermore:

\begin{proposition}
\label{prop:spin}
If all cone angles of $\SI$ are odd multiples of $2\pi$, then there exists a canonical $H^1(\Sigma;\Z_2)$-equivariant bijection $\S(X)\to\S(\Sigma)$.
\end{proposition}
\begin{proof}
Let $\kappa\in Z^1(X;S^1)$ be a fixed discrete canonical bundle over $\Sigma$. For each $\lambda\in Z^1(X;S^1)$ such that $\lambda^2=\kappa$, we shall now construct a vector
field $V_\lambda$ on $\SI$ with zeroes of even index. Such a vector field is well-known to define a spin structure, or equivalently -- by Johnson's theorem~\cite{Joh} --
a quadratic form $q_\lambda$ on $H_1(\SI;\Z_2)$. The proof will be completed with the verification that two equivalent $\lambda$'s induce identical quadratic forms, and that
the assignment $[\lambda]\mapsto q_\lambda$ is $H^1(\SI;\Z_2)$-equivariant.

\begin{figure}[h]
\labellist\small\hair 2.5pt
\pinlabel {$X$} at -50 150
\pinlabel {$X'$} at 830 150
\endlabellist
\centerline{\psfig{file=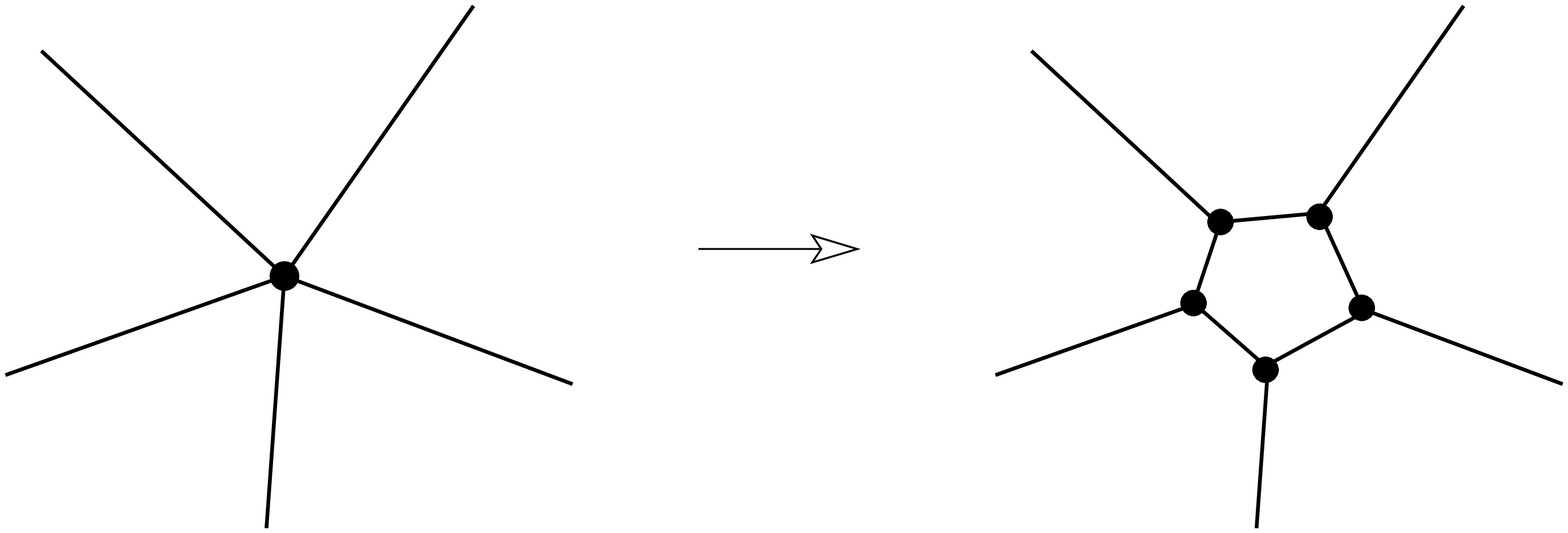,height=1.6cm}}
\end{figure}
First replace the cellular decomposition $X$ of $\SI$ by $X'$, where each singularity $v\in V(G)\cap S$ is
removed as illustrated above. Obviously, $\lambda$ induces $\lambda'\in Z^1(X';S^1)$ by setting $\lambda'(e)=1$ for each newly created edge $e$.
Fix an arbitrary orientation $\omega$ of the edges of $X'$. This allows to represent $\lambda'\in Z^1(X';S^1)$ as follows: write $\lambda'(e)=\exp(i\beta_{\lambda'}(e))$ with
$0\le\beta_{\lambda'}(e)<2\pi$ if $e$ is an edge oriented by $\omega$, and set $\beta_{\lambda'}(\bar{e})=-\beta_{\lambda'}(e)$ for the reverse edge.
Fix an arbitrary tangent vector $V_{\lambda}(v)$ at some arbitrary vertex $v$ of $X'$, and extend it to the 1-skeleton $G'$ of $X'$ as follows: running along an edge $e$,
rotate the tangent vector by an angle of $2\beta_{\lambda'}(e)$ in the negative direction. Since $\lambda$ is a cocycle such that $\lambda^2=\kappa$, and since each cone angle is a multiple of
$2\pi$, this gives a well-defined vector field along $G'$. Extend it to the whole surface $\SI$ by the cone construction, creating one zero in each face of $X$ and at each element of
$V(G)\cap S$. Obviously, the resulting vector field $V_\lambda$ depends on the choice of $\omega$, but not in a crucial way. Indeed, reversing the orientation $\omega$ on a given edge $e$ either does
nothing (if $\lambda'(e)=1$), or corresponds to adding two full twists to the vector field along $e$. Therefore, the parity of winding numbers with respect $V_\lambda$ are independent of $\omega$.
In particular, one easily checks that a zero of $V_\lambda$ is of even index if and only if the corresponding cone angle is an odd
multiple of $2\pi$, which we assumed.

As explained in~\cite{Joh}, the quadratic form $q_\lambda\colon H_1(\SI;\Z_2)\to\Z_2$ corresponding to $V_\lambda$ is determined as follows: for any regular oriented
simple closed curve $C\subset\Sigma\setminus S$, the number $q_\lambda([C])+1$ is equal to the mod 2 winding number of the tangential vector field along $C$ with
respect to the vector field $V_\lambda$. For an oriented simple closed curve $C\subset G$, we obtain the following equality modulo 2:
\[
q_\lambda([C])=1+\frac{1}{\pi}\Big(\sum_{e\subset C} \beta_\lambda(e)+\frac{1}{2}\sum_{v\in C}\alpha_v(C)\Big),
\]
where the first sum is over all oriented edges in the oriented curve $C$, and $\alpha_v(C)$ is the angle illustrated below. (This angle should be interpreted as explained in Remark~\ref{rem:alpha}.)
\begin{figure}[h]
\labellist\small\hair 2.5pt
\pinlabel {$v$} at 168 48
\pinlabel {$\alpha_v(C)$} at 300 100
\pinlabel {$C$} at 420 190
\endlabellist
\centerline{\psfig{file=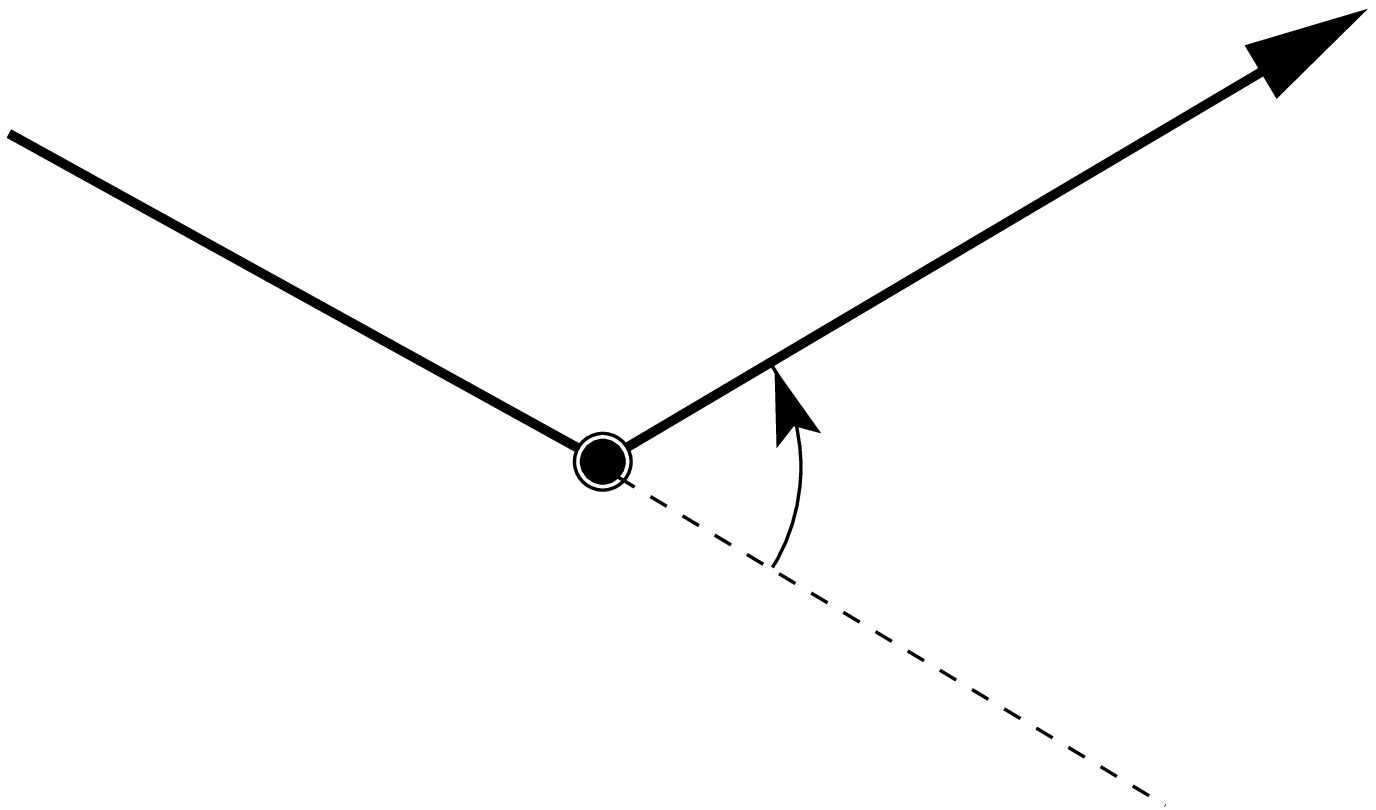,height=2cm}}
\end{figure}
Obviously, equivalent $\lambda$'s induce the same quadratic form $q_\lambda$. Finally, given two discrete spin structures $\lambda_1,\lambda_2$, the cohomology class of the 1-cocycle
$\lambda_1/\lambda_2\in Z^1(X;\{\pm 1\})$ is determined by its value on oriented simple closed curves in $G$. For such a curve $C$, we have
\[
(\lambda_1/\lambda_2)(C)=\exp\Big(i\sum_{e\subset C}(\beta_{\lambda_1}(e)-\beta_{\lambda_2}(e))\Big)=\exp\big(i\pi(q_{\lambda_1}-q_{\lambda_2})([C])\big).
\]
Therefore, the assignment $[\lambda]\mapsto q_\lambda$ is $H^1(\SI;\Z_2)$-equivariant, which concludes the proof.
\end{proof}

\subsection{The Kac-Ward formula for flat surfaces}

We are finally ready to prove the main result of this section, motivating the introduction of twisted Kac-Ward matrices and discrete spin structures.
As before, let $(G,x)$ be a weighted graph embedded in a closed orientable flat surface $\Sigma$ so that each edge of $G$ is a straight line, $\Sigma\setminus G$ consists of topological discs,
and the set $S$ of cone-type singularities is contained in $V(G)\cup V(G^*)$.

\begin{theorem}
\label{thm:Arf}
If all cone angles are odd multiple of $2\pi$, then the Kac-Ward determinant $\tau^\varphi(G,x)$ is the square of a polynomial in the $x_e$'s whenever $\varphi$ is a
discrete spin structure on $G\subset\SI$. Furthermore, if $\tau^\varphi(G,x)^{1/2}$ denotes the square root with constant coefficient equal to $+1$, then the Ising partition
function on $G$ is given by
\[
Z(G,x)=\frac{1}{2^g}\sum_{\lambda\in\S(X)}(-1)^{\Arf(\lambda)}\tau^\lambda(G,x)^{1/2},
\]
where $g$ is the genus of $\SI$ and $\Arf(\lambda)\in\Z_2$ the Arf invariant of the spin structure corresponding to $\lambda$.
\end{theorem}

The demonstration given below is by no means self-contained: it only consists in recasting the flat surface case in the more general (and more complicated) topological setting
discussed by the author in~\cite{Cim2}. We refer to this article for further details.

\begin{proof}
By Bass' Theorem \cite{Bas}, $\tau^\varphi(G,x)$ is given by
\[
\tau^\varphi(G,x)=\det(I-T^\varphi)=\prod_{\gamma\in\mathcal{P}(G)}\Big(1-\prod_{(e,e')\in\gamma}T^\varphi_{e,e'}\Big),
\]
where $\mathcal{P}(G)$ is the (infinite) set of prime reduced oriented closed paths in $G$, and the second product is over all pairs of consecutive oriented edges in the oriented path $\gamma$.
By definition of $T^\varphi$,
\[
\prod_{(e,e')\in\gamma}T^\varphi_{e,e'}=\varphi(\gamma)\,\exp\left(\textstyle{\frac{i}{2}}\alpha(\gamma)\right)\,x(\gamma),
\]
where $x(\gamma)=\prod_{e\in\gamma}x_e$ and $\alpha(\gamma)$ is the sum of the angles $\alpha(e,e')$ along $\gamma$ (interpreted as in Remark~\ref{rem:alpha}).
This equality already shows that $\tau^\varphi(G,x)$ does
not depend on the choice of the 1-cocycle representing the cohomology class $\varphi$. Furthermore, if $\varphi\in Z^1(X;S^1)$ is a discrete spin structure, then
\[
\left(\varphi(\gamma)\,\exp\left(\textstyle{\frac{i}{2}}\alpha(\gamma)\right)\right)^2=\varphi(\gamma)^2\,\exp(i\alpha(\gamma))=\varphi(\gamma)^2\,\kappa^{-1}(\gamma)=1
\]
for all oriented closed path $\gamma$ in $G$. This implies that $\varphi(\gamma)\,\exp\left(\frac{i}{2}\alpha(\gamma)\right)$ always belongs to
$\{\pm 1\}$. Therefore, for such a $\varphi=\lambda\in\S(X)$,
\[
\tau^\lambda(G,x)=\prod_{[\gamma]\in\mathcal{P}(G)/\sim}\left(1-(-1)^{w_\lambda({\gamma})}x(\gamma)\right)^2,
\]
for some $w_\lambda(\gamma)\in\Z_2$, where the equivalence relation on $\mathcal{P}(G)$ is given by $\gamma\sim-\gamma$. Using the notations of the proof of Proposition~\ref{prop:spin},
the element $w_\lambda(\gamma)$ satisfies
\[
(-1)^{w_\lambda(\gamma)}=\lambda(\gamma)\,\exp\left(\textstyle{\frac{i}{2}}\alpha(\gamma)\right)=
\exp\Big(i\sum_{e\subset\gamma}\beta_\lambda(e)+{\frac{i}{2}}\sum_{v\in\gamma}\alpha_v(\gamma)\Big).
\]
Therefore, $w_\lambda(\gamma)$ is nothing but the mod 2 winding number of the tangent vector field along $\gamma$ with respect to the vector field
$V_\lambda$ associated to the discrete spin structure $\lambda$. The formula now follows from Proposition~\ref{prop:spin} and \cite[Corollary~2.2]{Cim2}.
\end{proof}

\begin{remark}
The whole setting can be extended to encompass graphs embedded in flat surfaces with boundary. If there is exactly one boundary component, then Theorem~\ref{thm:Arf} extends verbatim.
(In particular, it applies to domains in the plane, where this theorem is exactly the original Kac-Ward formula~\cite{K-W}.)
If the flat surface has several boundary components, then the formula is slightly more complicated.
\end{remark}


\section{The Kac-Ward matrices with critical weights}
\label{sec:crit}

As proved in the previous section, the $\varphi$-twisted Kac-Ward matrices can be used to compute the Ising partition function for any weighted graph embedded in a flat surface.
We shall now assume the graph to be isoradial and the weights to be critical (recall Definitions~\ref{def:crit} and \ref{def:nu}). We will start with a combinatorial interpretation for
$\tau^\varphi(G,\nu)$ (Proposition~\ref{prop:tech}), that we then use for two of our main results. First, we prove a duality theorem relating $\tau^\varphi(G,\nu)$ and $\tau^\varphi(G^*,\nu^*)$
(Theorem~\ref{thm:duality}). Then, we show that $\tau^\varphi(G,\nu)$ coincides up to a multiplicative constant with the determinant of the critical discrete Laplacian
on $G$ if and only if the genus of $\Sigma$ is zero or one (Theorem~\ref{thm:Delta}). In a last subsection, we explain how $\tau^\varphi(G,\nu)$ can be understood as a discrete version of the
$\overline\partial$-torsion of the underlying Riemann surface.

\subsection{A combinatorial interpretation for $\tau^\varphi(G,\nu)$}
\label{sub:comb}

Let us start with some notations.

Given a graph $G$, let $\F(G)$ denote the set of subgraphs $F\subset G$ such that $F$ spans all vertices of $G$, and no connected component of $F$ is a tree.
Also, for a graph $F$ embedded in an oriented surface $\SI$, we shall denote by $N(F)$ a small tubular neighborhood of $F$ in $\SI$. (This is simply $F$ ``thickened", as illustrated in the middle
of Figure~\ref{fig:N}.) Since $\SI$ is oriented, so is $N(F)$, and this induces an orientation on the boundary $\partial N(F)$ of $N(F)$. Therefore, $\partial N(F)$ consists of a disjoint union of
oriented simple closed curves $\gamma$ on $\SI$. This is illustrated in Figure~\ref{fig:N}.

\begin{proposition}
\label{prop:tech}
Let $G$ be a graph isoradially embedded in a flat surface $\SI$, and let us assume that all cone angles $\vartheta_v$
of singularities $v\in V(G)$ are odd multiple of $2\pi$. Let $\nu_e=\tan(\theta_e/2)$ denote the critical weight system on $G\subset\Sigma$, and set $\mu_e=i\tan(\theta_e)$.
Then for any 1-cocycle $\varphi$,
\[
\tau^\varphi(G,\nu)=C\,\sum_{F\in\F(G)}\prod_{\gamma\subset\partial N(F)}(1-\varphi(\gamma))\,\mu(F),
\]
where the product is over all connected components $\gamma$ of the (clockwise oriented) boundary of a tubular neighborhood $N(F)$ of $F$ in $\Sigma$, $\mu(F)=\prod_{e\in F}\mu_e$,
and the constant $C$ is equal to
\[
C=(-1)^{|V(G)|}\,2^{-\chi(G)}\prod_{v\in V(G)}\exp(i\vartheta_v/4)\prod_{e\in E(G)}\frac{\cos(\theta_e)}{1+\cos(\theta_e)},
\]
with $\chi(G)=|V(G)|-|E(G)|$ the Euler characteristic of $G$.
\end{proposition}

\begin{figure}[Htb]
\labellist\small\hair 2.5pt
\pinlabel {$F$} at 110 330
\pinlabel {$N(F)$} at 690 345
\pinlabel {$\partial N(F)$} at 1300 360
\endlabellist
\centerline{\psfig{file=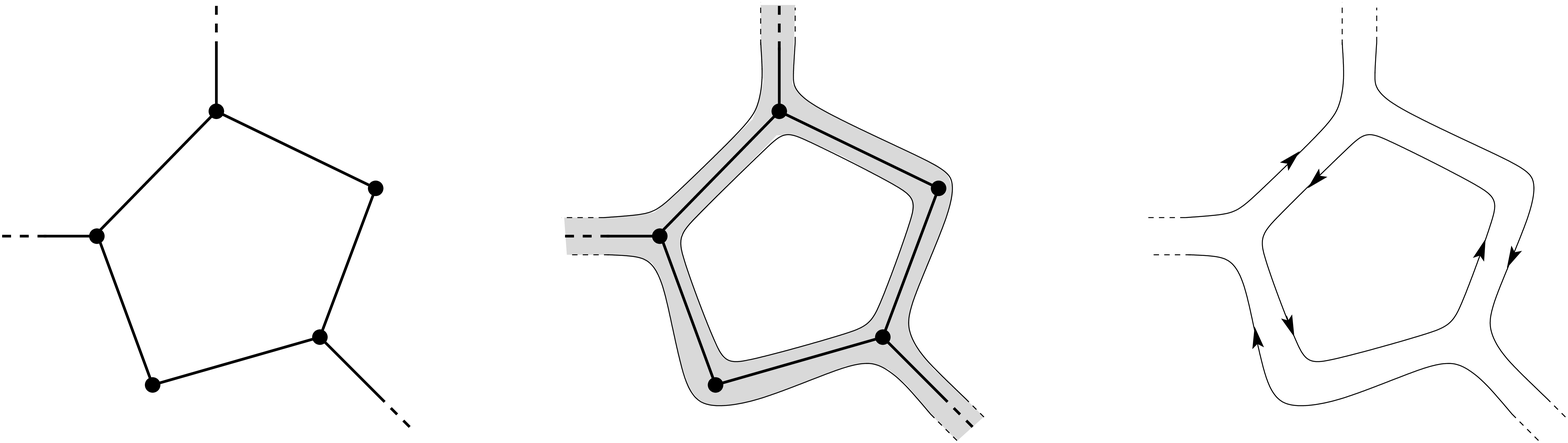,height=3.5cm}}
\caption{Part of a graph $F$ embedded in a surface $\SI$, the corresponding tubular neighborhood $N(F)$, and the oriented boundary $\partial N(F)$.}
\label{fig:N}
\end{figure}

The proof being quite substantial, we will split it into two lemmas. Let us begin with some notation. We shall write $V=V(G)$ and $E=E(G)$ for the sets of vertices and edges of the graph $G$,
${\EE}={\EE}(G)$ for the set of oriented edges of $G$ and $\L({\EE})$ for the complex vector space spanned by ${\EE}$. Obviously, ${\EE}$ can be partitioned into
${\EE}=\bigsqcup_{v\in V}E_v$, where $E_v$ contains all oriented edges $e$ with origin $o(e)=v$. Now, let us cyclically order the elements of $E_v$ by turning counterclockwise
around $v$. (As $\SI$ is orientable, this can be done in a consistent way.) Given $e\in E_v$, let $R(e)$ denote the next edge with respect to this cyclic order, as illustrated in Figure~\ref{fig:R}.
This induces an endomorphism $R$ of $\L({\EE})$. Also, let $J$ denote the endomorphism of $\L({\EE})$ given by $J(e)=\bar{e}$. Finally, we shall write $\mu$ for the endomorphism of $\L(\EE)$ given by
$\mu(e)=\mu_e\,e$, and similarly for any weight system and for $\varphi$. 

The aim of the first lemma is to relate $\tau^\varphi(G,\nu)$ to the determinant of a more tractable matrix.

\begin{lemma}
\label{lemma:KW}
For $G\subset\SI$ as in Proposition~\ref{prop:tech}, and for any 1-cocycle $\varphi$,
\[
\tau^\varphi(G,\nu)=2^{-\chi(G)}\prod_{e\in E}\frac{\cos^2(\theta_e)}{1+\cos(\theta_e)}\det M,
\]
with $M=I+J\varphi\mu-R(\mu+1)\in\mathit{End}(\L({\EE}))$ and $\chi(G)=|V(G)|-|E(G)|$.
\end{lemma}

\begin{proof}
Let $\mathit{Succ}\in\mathit{End}(\L({\EE}))$ be defined as follows: if $e$ is an oriented edge with terminus $t(e)=v$, then
\[
\mathit{Succ}(e)=i\varphi(e)\nu_e\sum_{e'\in E_v}\omega(e',\bar{e})\,e',
\]
where $\omega(e',\bar{e})=\exp(-\frac{i}{2}\beta(e',\bar{e}))$ for $e'\neq \bar{e}\in E_v$ with $\beta(e',\bar{e})$ as in Figure~\ref{fig:alpha}, and $\omega(\bar{e},\bar{e})=-1$. Also, let
$S\in\mathit{End}(\L({\EE}))$ be the endomorphism given by $S=\mathit{Succ}+iJ\varphi\nu$. By definition, the $\varphi$-twisted Kac-Ward matrix is the transposed of
$I-S$. Now, consider the matrix
\[
A=(I-S)(I+iJ\varphi\nu)=I-\mathit{Succ}+\mathit{Com},
\]
where
\[
\mathit{Com}(e)=-i\varphi(e)\nu_eS(\bar{e})=\nu_e^2\sum_{\genfrac{}{}{0pt}{}{e'\in E_v}{e'\neq e}}\omega(e',e)\,e'
\]
if $e$ has origin $o(e)=v$. Since
\[
\det(I+iJ\varphi\nu)=\prod_{e\in E}\det\left(\begin{array}{cc}1&i\varphi(\bar{e})\nu_e\cr i\varphi(e)\nu_e&1\cr\end{array}\right)=\prod_{e\in E}(1+\nu_e^2),
\]
we get the equality
\begin{equation}
\label{equ:A}
\tau^\varphi(G,\nu)=\prod_{e\in E}(1+\nu_e^2)^{-1}\det A.
\end{equation}
(The computation above is a variation on a trick due to Foata and Zeilberger, see \cite[Section 8]{F-Z}.)

\begin{figure}[Htb]
\labellist\small\hair 2.5pt
\pinlabel {$v$} at 50 270
\pinlabel {$e_1=R(e_d)$} at 390 90
\pinlabel {$e_2=R(e_1)$} at 445 350
\pinlabel {$e_3=R(e_2)$} at 260 510
\pinlabel {$e_4=R(e_3)$} at -30 500
\pinlabel {$\vdots$} at 5 280
\pinlabel {$e_d=R(e_{d-1})$} at 30 40
\pinlabel {$v$} at 780 260
\pinlabel {$e$} at 950 120
\pinlabel {$R(e)$} at 880 430
\pinlabel {$\theta_e$} at 940 215
\pinlabel {$\theta_{R(e)}$} at 960 300
\endlabellist
\centerline{\psfig{file=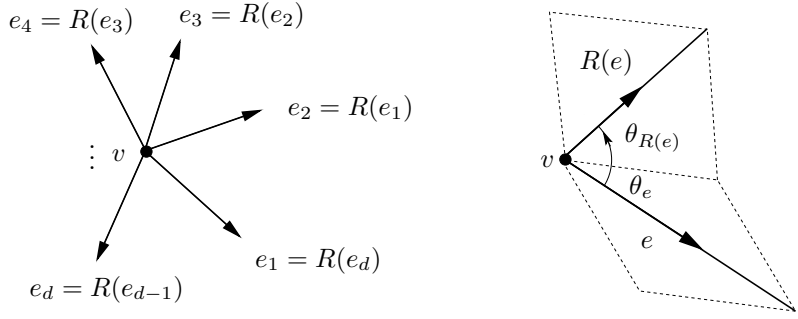,height=4cm}}
\caption{The endomorphism $R$ (to the left), and the equality $\beta(e,R(e))=\theta_e+\theta_{R(e)}$ (to the right).}
\label{fig:R}
\end{figure}

Let $Q\in\mathit{End}(\L({\EE}))$ be defined by
\[
Q(e)=\exp\left({\textstyle\frac{i}{2}}\beta(e,R(e))\right)e,
\]
and set $P=I-RQ$. Obviously, this endomorphism decomposes into $P=\bigoplus_{v\in V}P_v$ with $P_v\in\mathit{End}(\L({E}_v))$, and one easily computes
\[
\det P_v=1-\prod_{e\in E_v}\exp\left({\textstyle\frac{i}{2}}\beta(e,R(e))\right)=1-\exp\left({\textstyle\frac{i}{2}}\vartheta_v\right)=2,
\]
as the cone angle $\vartheta_v$ is an odd multiple of $2\pi$. Hence, the determinant of $P$ is $2^{|V|}$.
The point of introducing this $P$ is that it can be used to greatly simplify the matrix $A$. Indeed, let us compute the composition $PA$. If $e$ has terminus $v$, then
\begin{align*}
P\mathit{Succ}(e)&=(I-RQ)i\varphi(e)\nu_e\sum_{e'\in E_v}\omega(e',\bar{e})\,e'\cr
	&=i\varphi(e)\nu_e\sum_{e'\in E_v}\left(\omega(e',\bar{e})-\omega(R^{-1}(e'),\bar{e})\exp\left({\textstyle\frac{i}{2}}\beta(R^{-1}(e'),e')\right)\right)e'\cr
	&=-2i\varphi(e)\nu_e\,(\bar{e}).
\end{align*}
(To check that the coefficient of $e'=R(\bar{e})$ vanishes, we use once again the fact that $\vartheta_v$ is an odd multiple of $2\pi$.) Therefore, we have the equality
$P\mathit{Succ}=-2iJ\varphi\nu$. Similarly, given $e$ with $o(e)=v$,
\begin{align*}
P\mathit{Com}(e)=&(I-RQ)\nu_e^2\sum_{e'\in E_v\setminus\{e\}}\omega(e',e)\,e'\cr
	=&\nu_e^2\hskip-.2cm\sum_{e'\in E_v\setminus\{e,R(e)\}}\hskip-.2cm\left(\omega(e',e)-\omega(R^{-1}(e'),e)\exp\left({\textstyle\frac{i}{2}}\beta(R^{-1}(e'),e')\right)\right)e'\cr
	 &+\nu_e^2\left(\omega(R(e),e)\,R(e)-\omega(R^{-1}(e),e)\exp\left({\textstyle\frac{i}{2}}\beta(R^{-1}(e),e)\right)\,e\right)\cr
	=&-\nu^2_e(I+RQ)(e).
\end{align*}
These two equalities lead to
\begin{align*}
PA&=P(I-\mathit{Succ}+\mathit{Com})\cr
	&=(I-RQ)+2iJ\varphi\nu-(I+RQ)\nu^2\cr
	&=(1-\nu^2)I+2iJ\varphi\nu-RQ(1+\nu^2).
\end{align*}
We finally get
\begin{equation}
\label{equ:A'}
\det A=2^{-|V|}\prod_{e\in E}(1-\nu_e^2)^2\det A',
\end{equation}
where $A'$ is given by
\[
A'=I+iJ\varphi\,\frac{2\nu}{1-\nu^2}-RQ\,\frac{1+\nu^2}{1-\nu^2}.
\]

It is now time to use the fact that the weights $\nu_e$ are not any weights, but the critical ones given by $\nu_e=\tan(\theta_e/2)$. First observe that
\[
i\frac{2\nu_e}{1-\nu_e^2}=i\tan(\theta_e)=\mu_e\qquad\text{and}\qquad\frac{1+\nu_e^2}{1-\nu_e^2}=\frac{1}{\cos(\theta_e)}.
\]
Next, note the equality $\beta(e,R(e))=\theta_e+\theta_{R(e)}$ illustrated to the right of Figure~\ref{fig:R}. This implies that
$RQ(e)=\exp\left({\textstyle\frac{i}{2}}\theta_e\right)\exp\left({\textstyle\frac{i}{2}}\theta_{R(e)}\right)R(e)$. Multiplying each column of $A'$ (corresponding to $e$) by
$\exp\left({\textstyle\frac{i}{2}}\theta_e\right)$
and each line (corresponding to $e'$) by $\exp\left(-{\textstyle\frac{i}{2}}\theta_{e'}\right)$, we obtain a new matrix $M$ with
\begin{equation}
\label{equ:M}
\det A'=\det M\qquad\text{and}\qquad M=I+J\varphi\mu-R(1+\mu),
\end{equation} 
since $\theta_{\bar{e}}=\theta_e$ and $\frac{\exp(i\theta_e)}{\cos(\theta_e)}=1+\mu_e$. Equations~(\ref{equ:A}), (\ref{equ:A'}) and (\ref{equ:M}) give the statement of the lemma.
\end{proof}

In a second lemma, we now give a combinatorial interpretation of the determinant of this matrix $M$.

\begin{lemma}
\label{lemma:M}
The endomorphism $M=I+J\varphi\mu-R(1+\mu)$ of $\L({\EE})$ satisfies
\[
\det M=(-1)^{|V|}\prod_{e\in E}(1+\mu_e)\sum_{\genfrac{}{}{0pt}{}{F\subset G}{V(F)=V}}\prod_{\gamma\subset\partial N(F)}(1-\varphi(\gamma))\,\mu(F),
\]
the sum being on all subgraphs $F$ of $G$ spanning all vertices of $G$.
\end{lemma}

\begin{proof}
By definition, the coefficients of $M$ are given by
\[
M_{e,e'}=
\begin{cases}
1& \text{if $e=e'$;} \\
\varphi(e)\mu_e& \text{if $e'=\bar{e}$;} \\
-(1+\mu_e)&\text{if $e'=R(e)$,}
\end{cases}
\]
and vanish otherwise. Let us compute directly the determinant of $M$ as
\[
\det M=\sum_{\sigma\in S({\EE})}(-1)^{\text{sgn}(\sigma)}\prod_{e\in{\EE}}M_{e,\sigma(e)}.
\]
Each permutation $\sigma\in S({\EE})$ decomposes into disjoint cycles, inducing a partition ${\EE}=\bigsqcup_j E_j(\sigma)$ into orbits of length $\ell_j(\sigma)$.
The corresponding contribution to the determinant is
\[
(-1)^{\text{sgn}(\sigma)}\prod_{e\in E}M_{e,\sigma(e)}=\prod_j(-1)^{\ell_j(\sigma)+1}\prod_{e\in E_j(\sigma)}M_{e,\sigma(e)}.
\]
Since $M_{e,e}$ is equal to $1$, the oriented edges that are fixed by $\sigma$ contribute a trivial factor $1$ to this product and the corresponding orbits can be removed.
By definition of $M$, a permutation $\sigma\in S({\EE})$ will have a non-zero contribution only if each of the $n(\sigma)$ remaining orbits forms a cycle of oriented edges of $G$ such
that each oriented edge $e$ is either followed by $J(e)=\bar{e}$ or by $R(e)$, and such that these cycles pass through each edge of $G$ at most twice, and if so, in opposite directions.
If $\G(G)$ denotes the set of such union of cycles, we have
\[
\det M=\sum_{\gamma\in\G(G)}(-1)^{n(\gamma)}\prod_j\prod_{e\in\gamma_j}-M_{e,\gamma_j(e)},
\]
where each $\gamma\in\G(G)$ is written as a union of cycles $\gamma=\bigcup_{j=1}^{n(\gamma)}\gamma_j$. Now, for any fixed $\gamma$, a given edge $e\in E$ will fall in one of the following five
categories:
\begin{romanlist}
\item{$e$ is covered by $\gamma$ in both directions, as part of a cycle of the form $(\dots,R^{-1}(e),e,\bar{e},R(\bar{e}),\dots)$; the corresponding contribution to the determinant of $M$ is
$-\varphi(e)\mu_e(1+\mu_e)$.}
\item{$e$ is covered in both directions by the cycle $(e,\bar{e})$, so the contribution is $-\mu_e^2$. (The minus sign comes from the contribution of this cycle to $n(\gamma)$.)}
\item{$e$ is covered by $\gamma$ in both directions, as part of cycles of the form $(\dots,R^{-1}(e),e,R(e),\dots,R^{-1}(\bar{e}),\bar{e},R(\bar{e}),\dots)$;
in this case, the contribution is $(1+\mu_e)^2$.}
\item{$e$ is only covered in one direction, so the contribution is $(1+\mu_e)$.}
\item{$e$ is not covered at all, and the contribution is $1$.}
\end{romanlist}
For any given element $\gamma\in\Gamma(G)$, any edge of type $\ii$ in $\gamma$ can be removed (i.e. replaced by an edge of type $\mathit{(v)}$) and the resulting union of cycles will still belong to
$\Gamma(G)$. The converse also holds: any type $\mathit{(v)}$ edge in an element of $\Gamma(G)$ can be replaced by a type $\ii$ edge, the result will be in $\Gamma(G)$. Therefore, each time
an edge appears as an edge of type $\ii$ of some element of $\Gamma(G)$, it also appears as an edge of type $\mathit{(v)}$ of some other element of $\Gamma(G)$ and vice versa.
Using the equality $(1-\mu_e^2)=(1+\mu_e)(1-\mu_e)$, we therefore can factor out a term $1+\mu_e$ for each $e\in E$, leading to
\[
\det M=\prod_{e\in E}(1+\mu_e)\sum_{\gamma\in\G(G)}(-1)^{n(\gamma)}\prod_{e\in\gamma\i}\varphi(e)(-\mu_e)\prod_{e\in\gamma\ii}\mu_e\prod_{e\in\gamma\iii}(1+\mu_e),
\]
where $\gamma\i$ (resp. $\gamma\ii,\gamma\iii$) denotes the set edges of type $\i$ (resp. $\ii,\iii$) of $\gamma$. Next, we wish to expand the last product above as
\[
\prod_{e\in\gamma\iii}(1+\mu_e)=\sum_{\gamma'\subset\gamma\iii}\prod_{e\in\gamma'}\mu_e.
\]
Using the notation
\[
\widetilde{\G}(G)=\{(\gamma,\gamma')\,|\,\gamma\in\G(G),\;\gamma'\subset\gamma\iii\},
\]
we get the equality
\[
\det M=\prod_{e\in E}(1+\mu_e)\sum_{\widetilde\gamma\in\widetilde\G(G)}(-1)^{n(\widetilde\gamma)}\varphi(\widetilde{\gamma})\prod_{e\in\mathit{Supp}(\widetilde{\gamma})}\mu_e,
\]
where $n(\widetilde{\gamma})=n(\gamma)+|\gamma\i|$, the support of $\widetilde{\gamma}$ is $\mathit{Supp}(\widetilde{\gamma})=\gamma\i\cup\gamma\ii\cup\gamma'$ and
$\varphi(\widetilde{\gamma})=\prod_{e\in\gamma\i}\varphi(e)$ for any $\widetilde{\gamma}\in\widetilde{\G}(G)$. In other words,
\begin{equation}
\label{equ:tilde}
\frac{\det M}{\prod_{e\in E}(1+\mu_e)}=\sum_{F\subset E}c_F^\varphi\,\mu(F),\quad\text{where }
c_F^\varphi=\sum_{\genfrac{}{}{0pt}{}{\widetilde{\gamma}\in\widetilde{\G}(G)}{\mathit{Supp}(\widetilde{\gamma})=F}}(-1)^{n(\widetilde\gamma)}\varphi(\widetilde{\gamma}).
\end{equation}

We shall now check that $c^\varphi_F$ vanishes whenever $F$ does not span all vertices of $G$. Let us assume that the vertex $v$ is not spanned by $F$.
Let $\gamma_v$ be the cycle $(e,R(e),R^2(e),\dots,R^{d-1}(e))$ with $e\in E_v$ and $d=|E_v|$. The set of
$\widetilde{\gamma}\in\widetilde{\G}(G)$ with $\mathit{Supp}(\widetilde{\gamma})=F$ is equal to the disjoint union of $\G_0$ and $\G_1$, where $\G_0$ (resp. $\G_1$) denotes the set of such 
$\widetilde{\gamma}=(\gamma,\gamma')$ with $\gamma$ containing (resp. not containing) $\gamma_v$.
As $v$ does not belong to $\mathit{Supp}(\widetilde{\gamma})$, the mapping $(\gamma,\gamma')\mapsto(\gamma\cup\gamma_v,\gamma')$ gives a well-defined bijection
$f\colon\G_0\to\G_1$. Since $n(f(\widetilde{\gamma}))=n(\widetilde{\gamma})+1$ and $\varphi(f(\widetilde{\gamma}))=\varphi(\widetilde{\gamma})$, we get
\[
c_F^\varphi=\sum_{\widetilde{\gamma}\in\G_0\sqcup\G_1}(-1)^{n(\widetilde\gamma)}\varphi(\widetilde{\gamma})=
\sum_{\widetilde{\gamma}\in\G_0}((-1)^{n(\widetilde{\gamma})}+(-1)^{n(f(\widetilde{\gamma}))})\varphi(\widetilde{\gamma})=0.
\]

So, let us assume that $F\subset E$ spans all vertices of $G$, and let $\{T_k\}_k$ denote the connected components of $F$ viewed as a subgraph of $G$. Given $\gamma\in\G(G)$, let $\gamma|_{T_k}$ denote
the restriction of the cycles composing $\gamma$ to the oriented edges of $T_k$. One easily checks that if $(\gamma,\gamma')\in\widetilde{\G}(G)$ has support equal to $F$, then $\gamma|_{T_k}$
belongs to $\G(T_k)$ and $\gamma|_{T_k}(\bullet)=\gamma(\bullet)\cap T_k$ for $\bullet\in\{\mathit{i,ii,iii}\}$. Hence, sending $(\gamma,\gamma')$ to $(\gamma|_{T_k},\gamma'\cap T_k)$
gives well-defined maps
\[
\pi_k\colon\{\widetilde{\gamma}\in\widetilde{\G}(G)\,|\,\mathit{Supp}(\widetilde{\gamma})=F\}\to\{\widetilde{\gamma}_k\in\widetilde{\G}(T_k)\,|\,\mathit{Supp}(\widetilde{\gamma}_k)=T_k\},
\]
which in turn induce
\[
\{\widetilde{\gamma}\in\widetilde{\G}(G)\,|\,\mathit{Supp}(\widetilde{\gamma})=F\}\to\prod_k\{\widetilde{\gamma}_k\in\widetilde{\G}(T_k)\,|\,\mathit{Supp}(\widetilde{\gamma}_k)=T_k\}.
\]
Using the fact that $F$ spans all vertices of $G$, one can check that this map is a bijection.
Since $n(\widetilde{\gamma})=\sum_kn(\pi_k(\widetilde{\gamma}))$ and $\varphi(\widetilde{\gamma})=\prod_k\varphi(\pi_k(\widetilde{\gamma}))$ for any $\widetilde{\gamma}$ with support
equal to $F$, it follows that
\begin{equation}
\label{equ:c}
c^\varphi_F=\prod_k c^\varphi_{T_k}\quad\text{where }
c_T^\varphi=\sum_{\genfrac{}{}{0pt}{}{\widetilde{\gamma}\in\widetilde{\G}(T)}{\mathit{Supp}(\widetilde{\gamma})=T}}(-1)^{n(\widetilde\gamma)}\varphi(\widetilde{\gamma}).
\end{equation}
Note that the value of $c^\varphi_T$ does not depend on $G$ anymore. Furthermore, since $\mathit{Supp}(\widetilde{\gamma})$ must be equal to the whole of $T$, $\widetilde{\gamma}=(\gamma,\gamma')$
must satisfy $\gamma'=\gamma\iii$. Therefore,
\[
c_T^\varphi=\sum_{\genfrac{}{}{0pt}{}{\gamma\in\G(T)}{\mathit{Supp}(\gamma)=T}}(-1)^{n(\gamma)+|\gamma\i|}\varphi(\gamma),
\]
where $\mathit{Supp}(\gamma)=\gamma\i\cup\gamma\ii\cup\gamma\iii$ and $\varphi(\gamma)=\prod_{e\in\gamma\i}\varphi(e)$.

Let $N^\varphi_T\in\mathit{End}(\L({\EE}(T))$ be the endomorphism given by $N^\varphi_T=R-J\varphi$. Developing the determinant of $N^\varphi_T$ explicitely, we get
\[
\det N^\varphi_T=\sum_{\genfrac{}{}{0pt}{}{\gamma\in\G(T)}{\mathit{Supp}(\gamma)=T}}(-1)^{n(\gamma)+|\gamma\i|}\varphi(\gamma)=c^\varphi_T.
\]
As before, the set ${\EE}(T)$ is equal to $\bigsqcup_{v\in V(T)}E_v(T)$, and the endomorphism $R$ splits into $R=\bigoplus_{v\in V(T)}R_v$. Since $\det R_v=(-1)^{|E_v(T)|+1}$, we get
\[
c_T^\varphi=\det(R)\det(I-R^{-1}J\varphi)=(-1)^{|V(T)|}\det(I-R^{-1}J\varphi).
\]
Now, the orbits $E_j$ of the action of $R^{-1}J$ on the set ${\EE}(T)$ correspond exactly to the connected components of the oriented boundary of a tubular neighborhood $N(T)$ of $T$ in $\SI$,
where $N(T)$ is endowed with the clockwise orientation. Hence
\begin{equation}
\label{equ:N}
c_T^\varphi=(-1)^{|V(T)|}\prod_j\Big(1-\prod_{e\in E_j}\varphi(e)\Big)=(-1)^{|V(T)|}\prod_{\gamma\subset\partial N(T)}(1-\varphi(\gamma)).
\end{equation}
Since $|V|=|V(F)|=\sum_k|V(T_k)|$ and $N(F)=\bigsqcup_k N(T_k)$, Equations (\ref{equ:tilde}), (\ref{equ:c}) and (\ref{equ:N}) give the statement of the lemma.
\end{proof}

\begin{proof}[Proof of Proposition~\ref{prop:tech}]
Lemmas~\ref{lemma:KW} and \ref{lemma:M} give the equality
\[
\tau^\varphi(G,\nu)=C\,\sum_{\genfrac{}{}{0pt}{}{F\subset G}{V(F)=V}}\prod_{\gamma\subset\partial N(F)}(1-\varphi(\gamma))\,\mu(F),
\]
where the constant $C$ is equal to
\begin{align*}
C&=(-1)^{|V|}2^{-\chi(G)}\prod_{e\in E}\frac{\cos^2(\theta_e)}{1+\cos(\theta_e)}(1+i\tan(\theta_e))\cr
	&=(-1)^{|V|}2^{-\chi(G)}\prod_{e\in E}\frac{\cos(\theta_e)}{1+\cos(\theta_e)}\exp(i\theta_e).
\end{align*}
This is equal to the value given in the statement of the proposition, since
\[
\prod_{e\in E}\exp(i\theta_e)=\prod_{v\in V}\prod_{e\in E_v}\exp(i\theta_e/2)=\prod_{v\in V}\exp(i\vartheta_v/4).
\]
Finally, note that if a component of $F$ is a tree, then $\partial N(F)$ will contain a trivial cycle and the corresponding coefficient will vanish. Therefore, we can sum over all elements
of the set $\F(G)$ defined at the beginning of the section. This concludes the proof of the proposition.
\end{proof}

\subsection{The duality theorem}
\label{sub:duality}

We are now ready to state and prove one of our main results.

\begin{theorem}
\label{thm:duality}
Let $G$ be a graph isoradially embedded in a flat surface $\SI$, and let $\nu$ be the critical weight system on $G$.
The dual graph $G^*$ is also isoradially embedded in $\SI$, and therefore admits a critical weight system $\nu^*$. If all cone angles are odd multiples of $2\pi$, then for any $\varphi\in H^1(\SI;S^1)$,
\[
2^{|V(G^*)|}\hskip-2pt\prod_{e^*\in E(G^*)}(1+\cos(\theta_{e^*}))\,\tau^\varphi(G^*,\nu^*)=2^{|V(G)|}\hskip-2pt\prod_{e\in E(G)}(1+\cos(\theta_{e}))\,\tau^\varphi(G,\nu).
\]
\end{theorem}
\begin{proof}
Given a subgraph $F$ of $G$, let $\psi(F)$ denote the subgraph of $G^*$ given by
\[
\psi(F)=\{e^*\in E(G^*)\,|\,e\notin F\}.
\]
Obviously, $\psi$ defines a bijection from the set of subgraphs of $G$ onto the set of subgraphs of $G^*$, with inverse $\psi^{-1}(F^*)=\{e\in E(G)\,|\,e^*\notin F^*\}$. Let $\widetilde{\F}(G)$
denote the set of subgraphs of $G$ spanning all vertices of $G$, and containing no cycle that is the boundary of a face. These two conditions being dual to each other, the map $\psi$ defines a bijection
$\psi\colon\widetilde{\F}(G)\to\widetilde{\F}(G^*)$. Now, consider the sum
\[
D^\varphi(G,\mu)=\sum_{F\in\widetilde{\F}(G)}c_F^\varphi\,\mu(F),\quad\text{with } c_F^\varphi=\prod_{\gamma\subset\partial N(F)}(1-\varphi(\gamma)).
\]
Since dual edges have rhombus half-angles related by $\theta_{e^*}=\frac{\pi}{2}-\theta_e$, the weight $\mu_e=i\tan(\theta_e)$ satisfies $\mu_{e^*}=-\mu_e^{-1}$. Therefore,
\[
\prod_{e\in E(G)}\mu_e^{-1} D^\varphi(G,\mu)=\sum_{F\in\widetilde{\F}(G)}c_F^\varphi\,\mu^{-1}(\psi(F))=\hskip-2.2pt\sum_{F^*\in\widetilde{\F}(G^*)}(-1)^{|E(F^*)|}c_F^\varphi\,\mu^*(F^*).
\]
Furthermore, for any $F\in\widetilde{\F}(G)$, one can decompose the surface $\SI$ as the union of two tubular neighborhoods $N(F)$ and $N(\psi(F))$ of $F$ and $\psi(F)$, pasted along their
common boundary. (This is illutrated in Figure~\ref{fig:dec}.)
\begin{figure}[Htb]
\centerline{\psfig{file=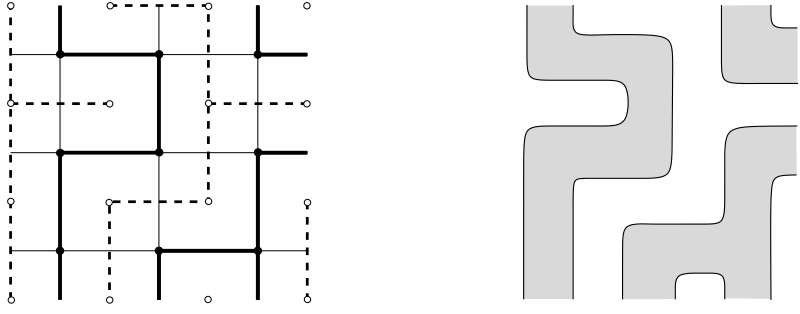,height=4cm}}
\caption{A graph $F$ (in solid lines) together with the associated graph $\psi(F)$ (in dashed lines), and the corresponding decomposition $\SI=N(F)\cup N(\psi(F))$.}
\label{fig:dec}
\end{figure}
In particular, the oriented boundaries satisfy $\partial N(F)=-\partial N(\psi(F))$. Therefore, since $\varphi$ is a 1-cocycle,
\begin{align*}
c_{F^*}^\varphi=&\prod_{\gamma^*\subset\partial N(F^*)}(1-\varphi(\gamma^*))=\prod_{\gamma\subset\partial N(F)}(1-\varphi^{-1}(\gamma))\cr
	=&(-1)^{|\partial N(F)|}\prod_{\gamma\subset\partial N(F)}\varphi^{-1}(\gamma)\,c_F^\varphi=(-1)^{|\partial N(F)|}\varphi^{-1}(\partial N(F))\,c_F^\varphi\cr
	=&(-1)^{|\partial N(F)|}c_F^\varphi,
\end{align*}
where $|\partial N(F)|$ denotes the number of boundary components of $N(F)$. Note also that, since $F$ spans all vertices of $G$, adding or removing an edge to $F$ changes the parity
of $|E(F^*)|$ and of $|\partial N(F)|$. Therefore, the parity of their sum does not depend on $F$, and it is easily seen to be equal to the parity of $|V(G^*)|$, the number of faces of $G\subset\SI$.
We have proved:
\[
\prod_{e\in E(G)}\mu_e^{-1} D^\varphi(G,\mu)=(-1)^{|V(G^*)|}D^\varphi(G^*,\mu^*).
\]
Since $c_F^\varphi$ vanishes whenever $F\in\widetilde{\F}(G)$ does not belong to $\F(G)$, Proposition~\ref{prop:tech} gives the equalities
\begin{align*}
\tau=&(-1)^{|V(G)|}2^{-\chi(G)}\prod_{v\in V(G)}\exp(i\vartheta_v/4)\prod_{e\in E(G)}\frac{\cos(\theta_e)}{1+\cos(\theta_e)}D^\varphi(G,\mu)\cr
\tau^*=&(-1)^{|V(G^*)|}2^{-\chi(G^*)}\prod_{v\in V(G^*)}\exp(i\vartheta_{v^*}/4)\prod_{e^*\in E(G^*)}\frac{\cos(\theta_{e^*})}{1+\cos(\theta_{e^*})}D^\varphi(G^*,\mu^*),
\end{align*}
where $\tau$ and $\tau^*$ stand for $\tau^\varphi(G,\nu)$ and $\tau^\varphi(G^*,\nu^*)$, respectively. The three equations above yield the equality
\[
2^{|V(G^*)|}\prod_{e^*\in E(G^*)}(1+\cos(\theta_{e^*}))\,\tau^*=A\cdot2^{|V(G)|}\prod_{e\in E(G)}(1+\cos(\theta_{e}))\,\tau,
\]
where the constant $A$ is equal to
\begin{align*}
A=& (-1)^{|V(G)|}\frac{\prod_{v}\exp(i\vartheta_v/4)}{\prod_{v^*}\exp(i\vartheta_{v^*}/4)}\prod_{e}\frac{\cos(\theta_e)}{\cos(\theta_{e^*})}i\tan(\theta_e)\cr
	=&(-1)^{|V(G)|+|V(G^*)|}\prod_{v}\exp(i\vartheta_v/4)\prod_{v^*}\exp(i\vartheta_{v^*}/4)(-1)^{|E(G)|}i^{-|E(G)|} \cr
	=& i^{\chi(\SI)}(-1)^{N+N^*},
\end{align*}
with $N$ the number of $v\in V(G)$ such that $\vartheta_v/2\pi$ is congruent to $3$ modulo $4$ (and similarly for $N^*$). By the discrete Gauss-Bonnet formula (recall Subsection~\ref{sub:flat}),
\[
2\pi\chi(\SI)=\sum_v(2\pi-\vartheta_v)+\sum_{v^*}(2\pi-\vartheta_{v^*}).
\]
This means exactly that $N+N^*$ and $\frac{1}{2}\chi(\SI)$ have the same parity, so $A$ is equal to $1$ and the theorem is proved.
\end{proof}

\subsection{Kac-Ward matrices versus discrete Laplacians}

Quite surprisingly, if the surface $\SI$ has genus zero or one, then our Kac-Ward determinants with critical weights turn out to be proportional to the determinants of discrete critical Laplacians.
Let us briefly recall the definition and main properties of these objects before stating the precise result.

As first observed by Eckmann~\cite{Eck} (in a much more general context), the Laplace operator on the space of complex valued smooth functions on a Riemann surface admits a beautifully simple discretization. It is the operator $\Delta$ on $C^0(G;\C)=\C^{V(G)}$ given by
\[
(\Delta f)(v)=\sum_{e=(v,w)}x_e\,(f(v)-f(w)),
\]
for any $f\in\C^{V(G)}$, where the sum is over all oriented edges $e$ of the form $(v,w)$.
Its codimension-one minors are very useful, as they count the number of (weighted) spanning trees in $G$: this is Kirchhoff's
celebrated matrix tree theorem~\cite{Kir}. On the other hand, the determinant of $\Delta$ always vanishes.

This construction admits a straightforward generalization:

\begin{definition}
\label{def:Delta}
Let $(G,x)$ be any weighted graph, and $\varphi\colon\pi_1(G)\to S^1$ any representation. The associated {\em discrete Laplacian\/}
is the operator $\Delta^\varphi=\Delta^\varphi(G,x)$ on $\C^{V(G)}$ defined by
\[
(\Delta^\varphi f)(v)=\sum_{e=(v,w)}x_e\,\left(f(v)-f(w)\varphi(e)\right),
\]
for any $f\in\C^{V(G)}$, where the sum is over all oriented edges $e=(v,w)$.
\end{definition}

The determinant of this discrete Laplacian has a nice combinatorial interpretation, which goes back at least to Forman~\cite{For}.
We include a proof here for the sake of completeness.

Given a graph $G$, let $\F_1(G)$ denote the set of subgraphs $F\subset G$ such that
$F$ spans all the vertices of $G$, and each connected component $T$ of $F$ has a unique cycle (i.e: $|V(T)|=|E(T)|$).

\begin{proposition}[Forman~\cite{For}]
\label{prop:For}
For any weighted graph $(G,x)$ and any $\varphi\colon\pi_1(G)\to S^1$,
\[
\det\Delta^\varphi(G,x)=\sum_{F\in\F_1(G)}\prod_{T\subset F}(2-\varphi(T)-\varphi(T)^{-1})\,x(F),
\]
where the product is over all connected components $T$ of $F$, $x(F)=\prod_{e\in F} x_e$, and $\varphi(T)=\prod_{e\in C_T}\varphi(e)$ with $C_T$ the unique cycle in $T$ endowed with
an arbitrary orientation.
\end{proposition}
\begin{proof}
Evaluating directly the determinant of $\Delta^\varphi(G,x)$ as a sum over permutations of $V=V(G)$, we get
\[
\det\Delta^\varphi(G,x)=\sum_{\sigma\in S(V)}(-1)^{\text{sgn}(\sigma)}\prod_{v\in V}\Delta_{v,\sigma(v)}.
\]
Each permutation $\sigma\in S(V)$ decomposes into disjoint cycles, inducing a partition $V=\bigsqcup_j V_j(\sigma)$.
A permutation will have a non-zero contribution to the determinant only if these orbits correspond to the disjoint union of oriented simple closed curves
$\gamma=\bigsqcup_{j=1}^{|\gamma|}\gamma_j$ in $G$. If $\Lambda(G)$ denotes the set of such union of curves, we get
\[
\det\Delta^\varphi(G,x)=\sum_{\gamma\in\Lambda(G)}(-1)^{|\gamma|}\varphi(\gamma)x(\gamma)\prod_{v\notin\gamma}\sum_{e\ni v}x_e,
\]
with $x(\gamma)=\prod_{e\in\gamma}x_e$ and $\varphi(\gamma)=\prod_{e\in\gamma}\varphi(e)$.

Given any subset $W\subset V$, let $\mathcal{X}(W)$ denote the set of maps $X$ assigning to each
$v\in W$ an oriented edge $X(v)\in{\EE}$ with origin $v$, and set $x(X)=\prod_{v\in W}x_{W(v)}$. For any ``vector field" $X\in\mathcal{X}(V)$, let $\Lambda(X)$ denote the set of
disjoint unions of ``trajectories" of $X$: these are oriented simple closed curves $C$ in $G$ such that any oriented edge $e\in C$ with origin $o(e)$ satisfies $X(o(e))=e$. We have
\begin{align*}
\det\Delta^\varphi(G,x)&=\sum_{\gamma\in\Lambda(G)}(-1)^{|\gamma|}\varphi(\gamma)x(\gamma)\sum_{\widetilde{X}\in\mathcal{X}(V\setminus(V\cap\gamma))}x(\widetilde{X})\cr
	&=\sum_{X\in\mathcal{X}(V)}\sum_{\gamma\in\Lambda(X)}(-1)^{|\gamma|}\varphi(\gamma)\,x(X)\cr
	&=\sum_{X\in\mathcal{X}(V)}\prod_C(1-\varphi(C))\,x(X),
\end{align*}
the product being over all trajectories $C$ of $X$. (This is the original formula of Forman~\cite{For}.) The map assigning to each $X\in\mathcal{X}(V)$ the subgraph $F=\{X(v)\}_{v\in V}$ is surjective
onto the set $\F_1(G)$ defined above, and each $F\in\F_1(G)$ has $2^{|F|}$ preimages given by the possible choices of orientations of the cycle in each of the $|F|$ connected components of $F$.
The proposition now follows from the equality $(1-\varphi(C))(1-\varphi(-C))=2-\varphi(C)-\varphi(C)^{-1}$.
\end{proof}

Let us mention one more fact: if $G$ is a planar isoradial graph, then the corresponding critical weights for the discrete Laplacian are given by $c_e=\tan(\theta_e)$, where
$\theta_e$ denotes the half-rhombus angle~\cite{Ken}. Therefore, we shall refer to $\Delta^\varphi(G,c)$ as the {\em critical discrete Laplacian\/} on $G$.
Finally, we shall make the usual abuse of notation and denote by the same letter $\varphi$ a homomorphism $\pi_1(\Sigma)\to S^1$ and the induced homomorphism on $\pi_1(G)$ for $G\subset\SI$.

\begin{theorem}
\label{thm:Delta}
Let $G$ be a graph isoradially embedded in a flat surface $\SI$, and let us assume that all cone angles $\vartheta_v$
of singularities $v\in V(G)$ are odd multiples of $2\pi$. Let $\nu_e=\tan(\theta_e/2)$ denote the critical weight system on $G\subset\Sigma$, and set $c_e=\tan(\theta_e)$.
If the genus of $\SI$ is zero or one, then for any $\varphi\colon\pi_1(\Sigma)\to S^1$,
\[
\tau^\varphi(G,\nu)=(-1)^N\,2^{-\chi(G)}\prod_{e\in E(G)}\frac{\cos(\theta_e)}{1+\cos(\theta_{e})}\,\det\Delta^\varphi(G,c),
\]
where $N$ is the number of vertices $v\in V(G)$ such that $\vartheta_v/2\pi$ is congruent to $3$ modulo $4$ and $\chi(G)=|V(G)|-|E(G)|$.
On the other hand, the functions $\tau^\varphi(G,\nu)$ and $\det\Delta^\varphi(G,c)$ are never proportional if the genus of $\SI$ is greater or equal to two.
\end{theorem}
\begin{proof}
If $\varphi$ is trivial, then both sides of the equality vanish by Propositions~\ref{prop:tech} and \ref{prop:For}. In particular, the equality holds in the genus zero case,
so it can be assumed that the genus $g$ of $\SI$ is positive. Let $F$ be an element of the set $\F(G)$, that is, a spanning subgraph of $G$ such that no connected component of $F$
is a tree. For any connected component $T$ of $F$, let $c^\varphi_T$ denote the corresponding coefficient $c^\varphi_T=\prod_{\gamma\subset\partial N(F)}(1-\varphi(\gamma))$. First note that
$c_T^\varphi$ only depends on the homotopy type of $T$ in $\SI$. (Basically, one can deform $T$ continuously in $\SI$ without changing $c_T^\varphi$.) Therefore, it can be assumed that $T$ is a wedge
of $n$ circles, with $n=1-\chi(T)=1+|E(T)|-|V(T)|\ge 1$.
If $n$ is greater then $2g$, then these $n$ cycles are linearly dependant in $H_1(\SI;\Z)$. Via a homotopy of $(\SI,T)$, one can therefore
assume that one of these cycles is null-homologous, leading to $c_T^\varphi=0$. If $n$ is equal to $2g$, then either these cycles are linearly dependant in $H_1(\SI;\Z)$ and $c_T^\varphi$ vanishes
as above, or these cycles are independent in homology. In this case, $T$ induces a cellular decomposition of $\SI$, so $\partial N(T)$ is the boundary of $f$ faces with
\[
2-2g=\chi(\SI)=|V(T)|-|E(T)|+f=1-2g+f.
\]
Therefore, $f$ is equal to $1$, so $\partial N(T)$ is connected, hence null-homologous, and $c_T^\varphi$ vanishes in this case as well.
We have proved that $c_F^\varphi$ vanishes unless each connected component $T$ of $F$ satisfies $0\le |E(T)|-|V(T)| \le 2g-2$.

In the case of genus $1$, this shows that $c_F^\varphi$ vanishes unless $F$ belongs to the set $\F_1(G)$. For such an element $F$, the contribution of each connected component $T$ can be easily computed:
$\partial N(T)$ consists of two connected components homologous to $T$ and $-T$, so
\[
c_T^\varphi=(1-\varphi(T))(1-\varphi(-T))=2-\varphi(T)-\varphi(T)^{-1}.
\]
By Propositions~\ref{prop:tech} and \ref{prop:For}, we now have
\begin{align*}
\tau^\varphi(G,\nu)&=C\sum_{F\in\F(G)}\prod_{\gamma\subset\partial N(F)}(1-\varphi(\gamma))\,\mu(F)\cr
	&=C\sum_{F\in\F_1(G)}\prod_{T\subset F}(2-\varphi(T)-\varphi(T)^{-1})\,\mu(F)\cr
	&=C\,\det\Delta^\varphi(G,\mu)\cr
	&=C\,i^{|V|}\det\Delta^\varphi(G,c),
\end{align*}
since $\mu_e=ic_e$. The equality now easily follows from the explicit value for the constant $C$ given in Proposition~\ref{prop:tech}.

Let us finally assume that the genus of $\SI$ is greater or equal to two. By the argument above, it is enough to show that there is some element $F\in\F(G)\setminus\F_1(G)$ with non-zero
coefficient $c^\varphi_F$. In particular, we just need to find a spanning connected subgraph $T\subset G$ with $|E(T)|-|V(T)|=1$ and such that none of the boundary components of $N(T)$
is trivial in $H_1(\SI;\Z)$. Such a $T$ is obtained as follows: choose a spanning tree $T_0\subset G$ and add two edges of $G$ so that the resulting two cycles are linearly independent in $H_1(\SI;\Z)$
but have zero intersection number. (This is possible since $\SI$ has genus $g>1$ and $G$ induces a cellular decomposition of $\SI$.)
The resulting graph $T$ satisfies the conditions listed above, and the proof is completed.
\end{proof}

\begin{remark}
\label{rem:BdT}
Boutillier and de Tili\`ere obtained a similar result in \cite{BdT1,BdT2} (see also~\cite{dT}): given any graph $G$ isoradially embedded in the flat torus,
they relate $\det\Delta^\varphi(G,c)$ with the determinant of some $\varphi$-twisted Kasteleyn matrix of $(\G_G,\nu)$, where $\G_G$ is the graph associated to $G$ via some variation of
the Fisher correspondence~\cite{Fi2}.  Actually, using the methods developed in \cite[Subsection 4.3]{Cim2}, one can show that Corollary 12 of~\cite{BdT2} is equivalent to our
Theorem~\ref{thm:Delta} in the flat toric case (that is, when $g=1$ and $S$ is empty). In our opinion, our approach has several advantages. It is more general, as we allow singularities
and understand the higher genus while Boutillier-de Tili\`ere only deal with the flat toric case;
it is more natural, as we work on the same graph $G$ throughout without using one of many possible auxiliary graphs $\G_G$; and it is simpler, as our whole proof relies
solely on Propositions~\ref{prop:tech} and \ref{prop:For}. The demonstration of Boutillier-de Tili\`ere, on the other hand, is quite substantial and relies on highly non-trivial results of
Kenyon~\cite{Ken} and Kenyon-Okounkov~\cite{K-O}.
\end{remark}

We conclude this section with one last remark.
Let us assume that $G$ is isoradially embedded in the flat torus $\SI=\C/\Lambda$, with $\Lambda$ some lattice in $\C$, and let $\widetilde{G}$ be the corresponding $\Lambda$-periodic planar graph.
For $n\ge 1$, set $G_n=\widetilde{G}/n\Lambda$. The {\em free energy per fundamental domain\/} of the critical Z-invariant Ising model on $\widetilde{G}$ is defined by
\[
f^I=-\lim_{n\to\infty}\frac{1}{n^2}\log Z^J(G_n),
\]
with $J_e=\frac{1}{2}\log\left(\frac{1+\sin\theta_e}{\cos\theta_e}\right)$ as explained in Subsection~\ref{sub:Z}. The equality
\[
Z^J(G_n)=\Big(\prod_{e\in E(G)}\cosh(J_e)\Big)^{n^2}2^{n^2|V(G)|}Z(G_n,\nu)
\]
together with Theorems~\ref{thm:Arf} and~\ref{thm:Delta} imply
\[
f^I=-|V(G)|\frac{\log(2)}{2}-\frac{1}{2}\log{\textstyle\det_1}\Delta(G,c),
\]
where $\log\det_1\Delta(G,c)=\lim_{n\to\infty}\frac{1}{n^2}\log\det\Delta^\lambda(G_n,c)$ and $\lambda$ is one of the three non-trivial spin structures on the torus.
(On the torus, spin structures are canonically identified with $H^1(\SI;\Z_2)$, the trivial spin structure has Arf invariant $1$ and the three others have Arf invariant $0$.
The trivial spin structure does not contribute to the partition function by Proposition~\ref{prop:tech}, and it is a fact that the three other spin structures will have the same
contribution to the free energy of the model.)
The free energy $f^I$ was computed by Baxter in~\cite{Bax2}, and the normalized determinant by Kenyon in~\cite{Ken} -- even though the existence of the limit was not proved there.
One can check that these two results are related in the way displayed above. This gives a reality check to our computations, and allows to obtain any of these two results as a corollary of the other one.

\subsection{The Kac-Ward determinants as discrete $\overline\partial$-torsions}
\label{sub:RS}

In this last subsection, we wish to relate (in an informal way) the critical Kac-Ward determinants with the $\overline\partial$-torsions of the underlying Riemann surface.
Let us start by briefly recalling the definition of this invariant, in the special case relevant to us.

To any closed N-dimensional complex manifold $\SI$ endowed with a unitary representation $\varphi\colon\pi_1(\SI)\to U(n)$ and a Hermitian metric, Ray and Singer \cite{RS2} associate a sequence of
numbers $T_p(\SI,\varphi)$ with $p=0,1,\dots,N$. In the case of a Riemann surface ($N=1$), the numbers $T_0(\SI,\varphi)$ and $T_1(\SI,\varphi)$ coincide, leading to a single invariant
$T(\SI,\varphi)$. If $n=1$, then $\varphi$ is a unitary character (that is, an element of $H^1(\SI;S^1)$), and it induces a complex line bundle $L(\varphi)$ over $\SI$.
Since $\SI$ is endowed with a Hermitian metric, one can consider the associated Laplacian $\Delta$ on the space of smooth sections of $L(\varphi)$.
The {\em $\overline\partial$-torsion\/} $T(\SI,\varphi)$ is defined as the square root of the zeta-regularized determinant
of this Laplacian, that is,
\[
T(\SI,\varphi)=\exp\Big(-{\textstyle\frac{1}{2}}\,\zeta^\prime(0)\Big),
\]
where $\zeta(s)$ is the zeta function of the Laplacian $\Delta$. Of course, this depends on the choice of Hermitian metric. However, for any two non-trivial characters $\varphi,\varphi'$, the ratio
$T(\Sigma,\varphi)/T(\Sigma,\varphi')$ is independent of this choice \cite[Theorem 2.1]{RS2}.

Let us come back to the discrete setting. As before, let $G$ be a graph isoradially embedded in a flat surface $\SI$ with all cone angles at singularities in $V(G)$ being odd multiples of $2\pi$.
Recall that the flat metric defines a conformal structure on the underlying surface, so that $\SI$ is now a Riemann surface. Consider the following statement.

\medskip

{\em As a function of $\varphi\in H^1(\SI;S^1)$, the Kac-Ward determinant $\tau^\varphi(G,\nu)$ behaves -- up to a multiplicative constant -- as
a discrete version of $T(\SI,\varphi)^2$, the square of the corresponding $\overline\partial$-torsion.}

\medskip

We shall not give any proof of this vague statement, not even formulate a precise conjecture (although this is very tempting).
Instead, we shall simply give a list of evidences towards such a statement.

First of all, recall that given any two non-trivial characters $\varphi,\varphi'$, the ratio $T(\Sigma,\varphi)/T(\Sigma,\varphi')$ only depends on the Riemann surface $\Sigma$ and on $\varphi,\varphi'$,
but not on the Hermitian metric.
Therefore, one expects the corresponding ratios $\tau^\varphi(G,\nu)/\tau^{\varphi'}(G,\nu)$ to exhibit some independence of the choice of the graph $G$ embedded in the surface $\Sigma$.
As mentioned in the introduction, the papers~\cite{CSM1,CSM2} provide numerical evidences towards such a claim at the scaling limit (in the special case where $\varphi$ and $\varphi'$ are
discrete spin structures). Now, our duality result can be understood as a further step in that direction. Indeed, it is a direct consequence of Theorem~\ref{thm:duality} that the ratios
$\tau^\varphi(G,\nu)/\tau^{\varphi'}(G,\nu)$ and $\tau^\varphi(G^*,\nu^*)/\tau^{\varphi'}(G^*,\nu^*)$ actually {\em coincide\/}, without taking any scaling limit, for any two non-trivial characters
$\varphi,\varphi'$.

In the remainder of the discussion, we shall distinguish between the genus zero, the genus one, and the higher genus cases.

\medskip

\noindent{\bf The genus zero case.}
Whenever the character $\varphi$ is trivial, $\tau^\varphi(G,\nu)$ vanishes by Proposition~\ref{prop:tech} (actually, Lemma~\ref{lemma:KW} is enough).
Since the kernel of the continuous Laplacian has dimension one, $T(\SI,\varphi)$ is also equal to zero. In particular, $\tau^\varphi(G,\nu)$ and $T(\SI,\varphi)^2$
trivially coincide in the case of genus zero.

\medskip

\noindent{\bf The genus one case.}
In the toric case, Theorem~\ref{thm:Delta} shows that $\tau^\varphi(G,\nu)$ is proportional to the most `natural' discretization of $T(\SI,\varphi)^2$, that is,
the determinant of the critical discrete Laplacian.

\medskip

\noindent{\bf The case of genus $\mathbf{g>1}$.}
Let us now consider the case of a Riemann surface of genus $g>1$. It can be written $\SI=H/\Gamma$, where $H$ is the Poincar\'e upper half plane and $\Gamma$ a discrete subgroup of $\mathit{PSL}(2,\R)$.
Since $\Sigma$ is compact, each element $\gamma\in\pi_1(\Sigma)=\Gamma$ acts on $H$ via $\gamma(z)=\exp(2\rho_\gamma)\frac{z-z_0}{z-z_1}$ for some real fixed points $z_0,z_1$ and $\rho_\gamma>0$.
Given any unitary character $\varphi$ of $\pi_1(\SI)=\Gamma$, the corresponding $\overline\partial$-torsion satisfies
\[
T(\SI,\varphi)^2=C\cdot Z_\Gamma(1,\varphi),
\]
with $C$ a constant depending on the Hermitian metric and $Z_\Gamma(\sigma,\varphi)$ the Selberg zeta function defined by
\[
Z_\Gamma(\sigma,\varphi)=\prod_{\{\gamma\}}\prod_{k\ge 0}(1-\varphi(\gamma)\exp(-2\rho_\gamma(\sigma+k))).
\]
Here, the product is over all conjugacy classes of primitive elements of $\Gamma$, that is, elements that are not powers in $\Gamma$ \cite[Theorem 4.6]{RS2}.
As explained in~\cite{Sar}, there is a more geometric way to understand this zeta function: conjugacy classes of primitive elements of $\Gamma$ correspond to primitive closed geodesics on $\Sigma$,
and $2\rho_\gamma$ is the length $\ell_\gamma$ of the corresponding geodesic. Summing up (and neglecting the problems of convergence), we have the equality
\[
T(\SI,\varphi)^2=C\prod_{\gamma\in\mathcal{P}(\SI)}\prod_{k\ge 1}(1-\varphi(\gamma)\exp(-k\ell_\gamma)),\eqno{(\star)}
\]
where $\mathcal{P}(\SI)$ denotes the set of primitive closed geodesics in $\SI$.

Let us make the corresponding study on the discrete side. Given an oriented closed path $\gamma$ on a graph $G$, we will say that $\gamma$ is {\em reduced} if it never backtracks,
that is, if no oriented edge $e$ is immediately followed by the oriented edge $\bar{e}$. The oriented closed path $\gamma$ will be called {\em prime} if, when viewed as a cyclic word,
it cannot be expressed as the product $\delta^r$ of a given closed path $\delta$ for any $r\ge 2$. Finally, we shall denote by $\mathcal{P}(G)$ the set of prime reduced oriented closed paths
in a graph $G$. Let us now assume that $G$ is isoradially embedded in a flat surface $\SI$. As explained in the
proof of Theorem~\ref{thm:Arf}, Bass' Theorem \cite{Bas} implies the equality
\[
\tau^\varphi(G,\nu)=\prod_{\gamma\in\mathcal{P}(G)}\Big(1-\varphi(\gamma)\exp\left(\textstyle{\frac{i}{2}}\alpha_\gamma\right)\nu(\gamma)\Big),\eqno{(\star\star)}
\]
where $\nu(\gamma)=\prod_{e\in\gamma}\nu_e$ and $\alpha_\gamma$ is the sum of the angles $\alpha(e,e')$ along $\gamma$ (recall Figure~\ref{fig:alpha}).

The expression $(\star\star)$ can be understood as a discrete version of $(\star)$ in the following sense. Primitive closed geodesics on $\SI$ are replaced by homotopically non-trivial
reduced prime closed paths in $G\subset\SI$. The factor $\exp\left(\textstyle{\frac{i}{2}}\alpha_\gamma\right)$ measures to which extend $\gamma$ differs from a straight line
(i.e. a geodesic on $\SI$ with respect to the flat metric). And finally, the critical weight $\nu_e=\tan(\theta_e/2)$ plays the role of $\exp(-\ell_e)$, with $\ell_e$ the length of the edge $e$.

\bibliographystyle{plain}

\bibliography{Ising}

\end{document}